\documentclass[a4paper,11pt]{article}
\usepackage{fullpage}
\usepackage{amsthm, amsfonts,mathrsfs,amssymb,amsmath,clrscode,multirow,graphicx,amsmath,epic,eepic}
\usepackage{epstopdf}
\usepackage{bigstrut,url,subfigure}

\newcommand{\field}{\mathbb{F}}
\newcommand{\D}{\mathcal{D}}
\newcommand{\A}{\mathfrak{t}}
\newcommand{\B}{\mathfrak{u}}
\newcommand{\F}{\mathcal{F}}
\newcommand{\degen}{\unlhd}
\newcommand{\Real}{\mathbb R}
\newcommand{\Int}{\mathbb Z}
\newcommand{\mset}[1]{\bar #1}

\newcommand{\myvec}[1]{\mathbf #1}
\newcommand{\supp}{\mathrm{supp}}
\newcommand{\poly}{\mathrm{poly}}
\newcommand{\identi}{\mathrm{id}}
\newcommand{\myspan}[1]{\mathrm{span}\{#1\}}
\newcommand{\Nn}{\mathcal{N}}
\newcommand{\Nat}{\mathbb N}
\newcommand{\braket}[1]{\langle #1 \rangle}
\newtheorem{theorem}{Theorem}[section]
\newtheorem{proposition}{Proposition}[section]
\newtheorem{definition}{Definition}[section]
\newtheorem{lemma}{Lemma}[section]
\begin{document}

\author{%
Fran{\c c}ois Le Gall \\ 
     Department of Computer Science\\
     Graduate School of Information Science and Technology\\
    The University of Tokyo\\
    \texttt{legall@is.s.u-tokyo.ac.jp} }
\title{Powers of Tensors and Fast Matrix Multiplication}
\date{}
\maketitle

\maketitle
\begin{abstract}
This paper presents a method to analyze the powers of a given trilinear form (a special kind of algebraic constructions also called a tensor) and obtain upper bounds on the asymptotic complexity of matrix multiplication. Compared with existing approaches, this method is based on convex optimization, and thus has polynomial-time complexity. As an application, we use this method to study powers of the construction given by Coppersmith and Winograd [Journal of Symbolic Computation, 1990] and obtain the upper bound $\omega<2.3728639$ on the exponent of square matrix multiplication, which slightly improves the best known upper bound. 
\end{abstract}

\section{Introduction}
Matrix multiplication is one of the most fundamental tasks in mathematics and computer science.
While the product of two $n\times n$ matrices over a field can naturally be computed
in $O(n^3)$ arithmetic operations, Strassen showed in 1969 that
$O(n^{2.81})$ arithmetic operations are enough \cite{Strassen69}. 
The discovery of this algorithm for matrix multiplication
with subcubic complexity   
gave rise to a new area of research, 
where the central question is to determine the value of the 
exponent of square matrix multiplication, denoted~$\omega$, and defined 
as the minimal value such that two
$n\times n$ matrices over a field can 
be multiplied using $O(n^{\omega+\varepsilon})$ arithmetic operations for any $\varepsilon>0$.
It has been widely conjectured that $\omega=2$ and several conjectures in combinatorics and group theory, 
if true, would lead to this result \cite{Alon+13,Cohn+FOCS05,Cohn+SODA13,Coppersmith+90}. 
However, 
the best upper bound obtained so far is $\omega<2.38$, as we explain below. 

Coppersmith and Winograd \cite{Coppersmith+90} showed in 1987 that  $\omega<2.3754770$. 
Their approach can be described as follows. A trilinear form is, 
informally speaking, a three-dimensional array with coefficients 
in a field $\field$. For any trilinear form~$t$ one can define its 
border rank, denoted $\underline R(t)$, which is a positive integer characterizing 
the number of arithmetic operations needed to compute the form.
For any trilinear form~$t$ and any real number $\rho\in[2,3]$, 
one can define a real number $V_\rho(t)$, called the \emph{value}
of the trilinear form. The theory developed by Sch\"onhage \cite{Schonhage81}
shows that, for any $m\ge 1$ and any $\rho\in[2,3]$, the following
statement hold:
\begin{equation}\label{statement}
\left(V_\rho(t^{\otimes m})\right)^{1/m}\ge \underline{R}(t) \:\:\Longrightarrow\:\: \omega\le \rho.
\end{equation}
Here the notation $t^{\otimes m}$ represents 
the trilinear form obtained by taking the $n$-th tensor power of~$t$. 
Coppersmith and Winograd presented a specific trilinear form $\A$, obtained by
modifying a construction given earlier by Strassen 
\cite{StrassenFOCS86}, computed its border rank $\underline{R}(\A)$,
and introduced deep techniques to estimate the value $V_\rho(\A)$.
In particular, they showed how
a lower bound $\tilde V_{\rho}(\A)$ on $V_\rho(\A)$ can be obtained for 
any $\rho\in[2,3]$ by solving an optimization problem.
Solving this optimization problem, they obtained the upper bound 
$\omega<2.3871900$, via Statement (\ref{statement}) with $t=\A$ 
and $m=1$, by finding the smallest $\rho$ such that 
$\tilde V_{\rho}(\A)\ge \underline{R}(\A)$.
They then proceeded to study the tensor power $\A^{\otimes 2}$ and
showed that, despite several new technical difficulties, a similar approach can be 
used to reduce the computation of a lower bound $\tilde V_\rho(\A^{\otimes 2})$ 
on $V_\rho(\A^{\otimes 2})$ to solving
another optimization problem of several variables. 
They discovered that $\tilde V_\rho(\A^{\otimes 2})>[\tilde V_\rho(\A)]^2$, 
due to the fact that the analysis of $\A^{\otimes 2}$ was finer, thus giving
a better upper bound on $\omega$
via Statement (\ref{statement}) with $t=\A$ 
and $m=2$.
Solving numerically the new optimization problem, they obtained
the upper bound $\omega<2.3754770$.

In view of the improvement obtained by taking the second tensor power, a natural question 
was to investigate higher powers of the construction $\A$ by Coppersmith and Winograph.
Investigating the third power was explicitly mentioned as an open problem in \cite{Coppersmith+90}.
More that twenty years later, Stothers showed that, while the third power does not seem to 
lead to any improvement, the fourth power does give an improvement \cite{Stothers10} (see also \cite{Davie+13}).
The improvement was obtained again via Statement (\ref{statement}), by showing how to reduce the 
computation of $V_\rho(\A^{\otimes 4})$ to solving a non-convex optimization problem.
The upper bound $\omega<2.3736898$ was obtained in \cite{Davie+13,Stothers10} by finding numerically a
solution of this optimization problem. It was later discovered that that solution was not optimal, and the improved 
upper bound $\omega<2.3729269$ was given in~\cite{WilliamsSTOC12} 
by exhibiting a better solution of the same optimization problem. 
Independently, Vassilevska Williams~\cite{WilliamsSTOC12} 
constructed a powerful and general framework to analyze recursively powers of a class of trilinear forms,
including the trilinear form $\A$ by Coppersmith and Winograd, and showed how
to automatically reduce, for any form $t$ in this class 
and any integer $m\ge 2$,
the problem of obtaining lower bounds on $V_\rho(t^{\otimes m})$
to solving (in general non-convex) optimization problems.
The upper bound $\omega<2.3729$ was obtained~\cite{WilliamsFull}
by applying this framework with $t=\A$ and $m=8$, and numerically solving this optimization 
problem.\footnote{Note that, while the upper bound on  
$\omega$ obtained for the eighth power is stated as $\omega<2.3727$ 
in the conference version~\cite{WilliamsSTOC12},
the statement has been corrected to $\omega<2.3729$ in the most recent version \cite{WilliamsFull},
since the previous bound omitted some necessary constraints in the optimization problem. 
Our results confirm the value of the latter bound, and increase its precision.}
A natural question is to determine what bounds on $\omega$ can be obtained by studying $\A^{\otimes m}$ for $m>8$.
One may even hope that, when $m$ goes to infinity, the upper bound on $\omega$ goes to two. 
Unfortunately, this question can hardly be answered by this approach since the 
optimization problems are highly non-convex and become intractable even for modest values of~$m$.

In this paper we show how to modify the framework developed in \cite{WilliamsSTOC12} in such a way that
the computation of $V_\rho(\A^{\otimes m})$ reduces to solving $\poly(m)$ instances of \emph{convex} 
optimization problems, each having $\poly(m)$ variables.
From a theoretical point a view,
since a solution of such convex
problems can be found in polynomial time, via 
Statement~(\ref{statement}) we obtain an algorithm to derive an upper bound on~$\omega$ 
from $\A^{\otimes m}$ in time polynomial in $m$. 
From a practical point of view, the
convex problems we obtain can also be solved efficiently, and have several desirable properties
(in particular, the optimality of a solution can be guaranteed by using the dual problem).
We use this method to analyze $\A^{\otimes 16}$
and $\A^{\otimes 32}$,
and obtain the new upper bounds on~$\omega$ described in Table~\ref{table:chart}.
Besides leading to an improvement for~$\omega$, these results 
strongly suggest that studying powers higher than~32 
will give only negligible improvements.

Our method is actually more general and can be used to efficiently 
obtain lower bounds on $V_\rho(t\otimes t')$ for any trilinear forms $t$ and $t'$ that have
a structure ``similar" to~$\A$. 
Indeed, considering possible future applications of our approach, 
we have been attentive of stating our techniques as generally as possible.
To illustrate this point, we work out in the appendix the application of our method to an asymmetric trilinear form, 
originally proposed in \cite{Coppersmith+90}. \vspace{-2mm}

\begin{table}[h!]
\renewcommand\arraystretch{1}
\begin{center}
\caption{Upper bounds on $\omega$ obtained by analyzing the $m$-th power of the 
construction $\A$ by Coppersmith and Winograd.}\label{table:chart}\vspace{3mm}
\begin{tabular}{|c|l|l|}
\hline
$m$& Upper bound& Reference\bigstrut\\ 
\hline
1& $\omega<2.3871900$ & Ref.~\cite{Coppersmith+90}\\
\hline
2&  $\omega<2.3754770$ & Ref.~\cite{Coppersmith+90}\\
\hline
4& $\omega<2.3729269$& Ref.~\cite{WilliamsSTOC12}\\
\hline
\multirow{2}{*}{8}& $\multirow{2}{*}{$\omega<2.3728642$}$ & this paper (Section \ref{sub:mixed})\\
&& ($\omega<2.3729$ given in Ref.~\cite{WilliamsFull})\\
\hline
16& $\omega<2.3728640$ & this paper (Section \ref{sub:mixed})\\
\hline
32& $\omega<2.3728639$ & this paper (Section \ref{sub:mixed})
\tabularnewline\hline
\end{tabular}
\end{center}
\end{table}\vspace{-3mm}

\section{Algebraic Complexity Theory}
This section presents the notions of algebraic complexity needed for this work.
We refer to, e.g.,  \cite{Blaser13, Burgisser+97} for more detailed treatments.
In this paper $\field$ denotes an arbitrary field. 

\subsection{Trilinear forms}
Let $u,v$ and $w$ be three positive integers, and  $U$, $V$ and~$W$ be three vector spaces over
$\field$ of dimension $u$, $v$ and $w$, respectively.
A trilinear form (also called a tensor)~$t$ on $(U,V,W)$ is an element in $U\otimes V\otimes W\cong\field^{u\times v\times w}$,
where $\otimes$ denotes the tensor product.
If we fix bases $\{x_i\}$, $\{y_j\}$ and $\{z_k\}$ of $U$, $V$ and $W$, respectively,
then $t$ can be written as
$$
t=\sum_{i,j,k}t_{ijk}\:x_i\otimes y_j\otimes z_k
$$
for coefficients $t_{ijk}$ in $\field$. 
We will usually write $x_i\otimes y_j\otimes z_j$ simply as $x_iy_jz_k$. 

Matrix multiplication of an $m\times n$ matrix with entries in~$\field$ by an $n\times p$ matrix
with entries in~$\field$
corresponds to the trilinear form on $(\field^{m\times n},\field^{n\times p},\field^{m\times p})$ with coefficients 
$t_{ijk}=1$ if $i=(r,s)$, $j=(s,t)$ and $k=(r,t)$ for some integers $(r,s,t)\in \{1,\ldots,m\}\times \{1,\ldots,n\}\times \{1,\ldots,p\}$, and $t_{ijk}=0$ otherwise. Indeed, this form can be rewritten as
\[
\sum_{r=1}^m\sum_{t=1}^p\left(\sum_{s=1}^n x_{(r,s)}y_{(s,t)}\right)z_{(r,t)}.
\]
Then, replacing the $x$-variables by the entries of the first matrix and the $y$-variables by the entries of the second matrix, the coefficient of $z_{(r,t)}$ in the above expression represents the entry in the $r$-th row and the
$t$-th column of the matrix product of these two matrices.
This trilinear form will be denoted by $\braket{m,n,p}$.

Another important example is the form $\sum_{\ell=1}^n x_\ell y_\ell z_\ell$. 
This trilinear form on $(\field^n,\field^n,\field^n)$ is denoted~$\braket{n}$ and corresponds to $n$ independent scalar products.

Given a tensor $t\in U\otimes V\otimes W$, it will be convenient to 
denote by $t_{\mathsf C}$ and $t_{\mathsf C^2}$ the tensors in $V\otimes W\otimes U$
and $W\otimes U\otimes V$, respectively, obtained by permuting cyclicly the coordinates
of $t$:
\[
t_{\mathsf C}=\sum_{ijk}t_{ijk}\:y_j\otimes z_k\otimes x_i,\hspace{5mm}t_{\mathsf C^2}=\sum_{ijk}t_{ijk}\:z_k\otimes x_i\otimes y_j.
\]

Given two tensors 
$t\in U\otimes V\otimes W$ and $t'\in U'\otimes V'\otimes W'$,
we can naturally define their direct sum $t\oplus t'$, which is a tensor in 
$(U\oplus U')\otimes (V\oplus V')\otimes (W\oplus W')$, and their tensor product 
$t\otimes t'$, which is a tensor in 
$(U\otimes U')\otimes (V\otimes V')\otimes (W\otimes W')$. 
For any integer $c\ge 1$, the tensor $t\oplus\cdots\oplus t$ (with~$c$ occurrences of $t$) 
will be denoted
by $c\cdot t$ and 
the tensor $t\otimes\cdots \otimes t$ (with $c$ occurrences of $t$)
will be denoted by $t^{\otimes c}$.

Let $\lambda$ be an indeterminate and consider the extension
$\field[\lambda]$ of $\field$, i.e., the set of all polynomials over $\field$ in $\lambda$. 
Let $t\in \field^{u\times v\times w}$ and $t'\in \field^{u'\times v'\times w'}$ be two tensors.
We say that~$t'$ is a degeneration of $t$, denoted $t'\degen t$,  if there exist three matrices
$
\alpha\in \field[\lambda]^{u'\times u},\:\:\beta\in \field[\lambda]^{v'\times v},\:\:\gamma\in \field[\lambda]^{w'\times w}
$
such that 
\[
\lambda^st'+\lambda^{s+1} t''=\sum_{ijk}t_{ijk}\:\alpha(x_i)\otimes \beta(y_j)\otimes \gamma(z_j)
\]
for some tensor $t''\in \field[\lambda]^{u'\times v'\times w'}$ and some nonnegative integer $s$.
Intuitively, the fact that a tensor $t'$ is a degeneration of a tensor $t$ means that an 
algorithm computing~$t$ can be converted into another algorithm computing $t'$ 
with essentially the same complexity.
The notion of degeneration can be used to define the notion of border rank of a tensor $t$, 
denoted $\underline{R}(t)$, as follows:
\[
\underline{R}(t)=\min\{{r\in\Nat\:|\: t\degen\braket{r}\}}.
\]
The border rank is submultiplicative: $\underline{R}(t\otimes t')\le \underline{R}(t)\times \underline{R}(t')$
for any two tensors $t$ and $t'$.

\subsection{The exponent of matrix multiplication}
%
The following theorem, which was proven by Sch\"onhage~\cite{Schonhage81}, 
shows that good upper bounds on $\omega$ can be obtained by finding a trilinear form 
of small border rank that can be degenerated into a direct sum of several large matrix multiplications.  

\begin{theorem}\label{theorem_schonhage}
Let $e$ and $m$ be two positive integers.
Let~$t$ be a tensor such that $e\cdot \braket{m,m,m}\unlhd t$.
Then 
$
em^{\omega}\le \underline{R}(t)
$.
\end{theorem}

Our results will require a generalization of Theorem \ref{theorem_schonhage},
based on the concept of \emph{value} of a tensor. Our presentation of this concept
follows \cite{Davie+13}.
Given a tensor $t\in \field^{u\times v\times w}$ and a positive integer $N$,
define the set 
\begin{equation}\label{eq:value}
\left\{(e,m)\in \Nat\times \Nat\:|\:e\!\cdot\! \braket{m,m,m}\unlhd (t\otimes t_{\mathsf C}\otimes t_{\mathsf C^2})^{\otimes N}\right\}
\end{equation}
corresponding to all pairs $(e,m)$ such that the tensor $(t\otimes t_{\mathsf C}\otimes t_{\mathsf C^2})^{\otimes N}$ can be degenerated
into a direct sum of $e$ tensors, each isomorphic to $\braket{m,m,m}$. Note that this set is finite.
For any real number $\rho\in[2,3]$,
define 
\[
V_{\rho,N}(t)=\max\{(em^\rho)^{\frac{1}{3N}}\},
\]
where the maximum is over all $(e,m)$ in the set of Eq.~(\ref{eq:value}).
We now give the formal definition of the value of a tensor.
\begin{definition}
For any tensor $t$ and any $\rho\in[2,3]$, 
\[
V_{\rho}(t)=\lim_{N\to\infty}V_{\rho,N}(t).
\]
\end{definition}
The limit in this definition is well defined, see \cite{Davie+13}.
Obviously, $V_\rho(t_{\mathsf C})=V_\rho(t_{\mathsf C^2})=V_\rho(t)$ for any tensor $t$,
and $V_\rho(t)\ge V_\rho(t')$ for any tensors $t,t'$ such that 
$t'\degen t$. 
By definition, for any positive integers $m,n$ and $p$ we have
\[
V_\rho(\braket{m,n,p})\ge (mnp)^{\rho/3}.
\]
Moreover,
the value is superadditive and supermultiplicative: for any two tensors
$t$ and $t'$, and any $\rho\in[2,3]$, the inequalities
$V_\rho(t\oplus t')\ge V_\rho(t)+V_\tau(t')$ and
$V_\rho(t\otimes t')\ge V_\rho(t)\times V_\tau(t')$
hold.
With this concept of value, we can state the following slight 
generalization of Theorem \ref{theorem_schonhage}, which 
was used implicitly in \cite{Coppersmith+90} and stated 
explicitly in \cite{Davie+13,WilliamsSTOC12}. 
\begin{theorem}\label{th:value}
Let $t$ be a tensor and $\rho$ be a real number 
such that $2\le \rho\le 3$. 
If
$
V_\rho(t)\ge \underline{R}(t),
$
then $\omega\le \rho$.
\end{theorem}

Finally, we will need the concept of decomposition of a tensor.
Our presentation of this concept
follows \cite{Burgisser+97}.
Let $t\in U\otimes V\otimes W$ be a tensor. Suppose that the vector spaces $U$, $V$ and $W$
decompose as 
\[
U=\bigoplus_{i\in I} U_i,\:\:V=\bigoplus_{j\in J} V_j,\:\:W=\bigoplus_{k\in K} V_k,
\]
where $I,J$ and $W$ are three finite subsets of $\Int$. Let us call this decomposition $D$.
We say that $D$ is a decomposition of $t$ if the tensor $t$ can be written as
\[
t=\sum_{(i,j,k)\in I\times J\times K} t(i,j,k),
\]
where each $t(i,j,k)$ is a tensor in $U_i\otimes V_j\otimes W_k$ (the sum 
does not need to be direct). The support of $t$ with respect to $D$
is defined as 
\[
\supp(t)=\{(i,j,k)\in I\times J\times K\:|\:t(i,j,k)\neq 0\},
\]
and the nonzero $t(i,j,k)$'s are called the components of $t$.

\section{Preliminaries and Notations}\label{sec:prelim}
In this section $S$ denotes a finite subset of $\Int\times \Int\times \Int$.

Let $\alpha_1,\alpha_2,\alpha_3\colon S\to \Int$ be the three coordinate
functions of~$S$, which means that $\alpha_\ell(\myvec{s})=s_\ell$ 
 for each $\ell\in\{1,2,3\}$ and all $\myvec{s}=(s_1,s_2,s_3)\in S$.
We first define the concept of tightness. The same notion was used in \cite{Burgisser+97}.
\begin{definition}
The set $S$ is tight if there exists an integer $d$ such that 
$\alpha_1(\myvec{s})+\alpha_2(\myvec{s})+\alpha_3(\myvec{s})=d$ 
for all $\myvec{s}\in S$. 
The set $S$ is $b$-tight, where $b$ is a positive integer,
if additionally
$\alpha_\ell(S)\subseteq \{0,1,\ldots,b-1\}$ for all $\ell\in\{1,2,3\}$.
\end{definition}
Note that if $S$ is $b$-tight then $|S|\le b^2$.

We denote by $\F(S)$ the set of all real-valued functions on~$S$,
and by $\D(S)$ the set of all probability distributions on~$S$
(i.e., the set of all functions $f\in \F(S)$ such that $f(\myvec{s})\ge 0$ for 
each $\myvec{s}\in S$ and $\sum_{\myvec{s}\in S}f(\myvec{s})=1$).
Note that, with pointwise addition and scalar multiplication, 
$\F(S)$ forms a real vector space of dimension $|S|$.
Given any function $f\in \F(S)$,
we denote by $f_1\colon\alpha_1(S)\to \Real$, $f_2\colon\alpha_2(S)\to \Real$ and $f_3\colon\alpha_3(S)\to \Real$
the three marginal functions of $f$: for each $\ell\in\{1,2,3\}$ and each $a\in \alpha_\ell(S)$,
\begin{align*}
f_\ell(a)&=\sum_{\myvec{s}\in \alpha_\ell^{-1}(a)}f(\myvec{s}).
\end{align*}

Let $\mathbb{S}_S$ denote the group of all permutations on $S$. 
Given any function $f\in \F(S)$ and any $\sigma\in \mathbb{S}_S$, we will denote by
$f^{\sigma}$ the function in $\F(S)$ such that
$f^{\sigma}(\myvec{s})=f(\sigma(\myvec{s}))$ for all $\myvec{s}\in S$. We now define 
the concept of invariance of a function.
\begin{definition}
Let $G$ be a subgroup of $\mathbb{S}_S$.
A function $f\in \F(S)$ 
is $G$-invariant if $f^\sigma=f$ for all $\sigma\in G$. 
\end{definition}
We will denote by 
$\F(S,G)$ the set of all $G$-invariant real-valued functions on $S$,
and by $\D(S,G)=\D(S)\cap \F(S,G)$
the set of all $G$-invariant probability distributions on~$S$.
We denote by $\F_0(S,G)$ the vector space of
all functions $f\in \F(S,G)$ such that $f_\ell(a)=0$ for all $\ell\in\{1,2,3\}$ and all $a\in \alpha_\ell(S)$,
and write 
\[
\chi(S,G)=\dim(\F_0(S,G)).
\]
We call this number $\chi(S,G)$ the compatibility degree of $S$ with respect to $G$.

In our applications  it will be sometimes more convenient to characterize the invariance
in term of a subgroup of permutations of the three coordinates of $S$, rather than 
in term of a subgroup of permutations on $S$, as follows.
Let $L$ be a subgroup of $\mathbb{S}_3$, the group of permutations over
$\{1,2,3\}$. We say that $S$
is $L$-symmetric if $(s_{\sigma(1)},s_{\sigma(2)},s_{\sigma(3)})\in S$ for all 
$(s_1,s_2,s_3)\in S$ and all $\sigma\in L$.
If $S$ is $L$-symmetric, the subgroup $L$
induces the subgroup 
$
L_S=\{\pi_\sigma\:|\:\sigma\in L\}
$
of $\mathbb{S}_S$, where $\pi_\sigma$ denotes the permutation in $\mathbb{S}_S$
such that 
$
\pi_\sigma(s_1,s_2,s_3)=(s_{\sigma(1)},s_{\sigma(2)},s_{\sigma(3)})
$ 
for 
all $(s_1,s_2,s_3)\in S$. 
We will slightly abuse notation and, when $S$ is $L$-symmetric, simply write 
$\F(S,L)$, $\D(S,L)$, $\F_0(S,L)$, $\chi(S,L)$ to represent  
$\F(S,L_S)$, $\D(S,L_S)$, $\F_0(S,L_S)$ and $\chi(S,L_S)$, respectively.

The entropy of a probability distribution $P\in \D(S)$ is 
\[
H(P)=-\sum_{\myvec{s}\in S}P(\myvec{s})\log(P(\myvec{s})),
\]
with the usual convention $0\times \log(0)=0$.
Using the above notations, 
$P_1$, $P_2$ and $P_3$ represent the three marginal
probability distributions of $P$.
For each $\ell\in\{1,2,3\}$, the entropy of  
$P_\ell$ is 
\[
H(P_\ell)=-\sum_{a\in \alpha_\ell(S)}P_\ell(a)\log(P_\ell(a)).
\]
It will sometimes be more convenient to represent, for each $\ell\in\{1,2,3\}$, 
the distribution $P_\ell$ 
as a vector $\myvec{P_\ell}\in\Real^{|\alpha_{\ell}(S)|}$, 
by fixing an arbitrary ordering of the elements in $\alpha_\ell(S)$.

We now define the concept of compatibility of two probability distributions.
\begin{definition}
Two probability distributions $P$ and $Q$ in $\D(S)$ are
compatible if $P_\ell=Q_\ell$ for each $\ell\in\{1,2,3\}$. 
\end{definition}
Finally, for any $P\in \D(S)$, 
we define the quantity
\[
\Gamma_S(P)=\max_{Q}\left[H(Q)\right] - H(P),
\]
where the maximum is over all $Q\in \D(S)$ compatible with~$P$.
Note that $\Gamma_S(P)$ is always non-negative.

\section{General Theory}\label{sec:main}
In this section we describe how to analyze the value of a trilinear form that has a decomposition
with tight support.
\subsection{Derivation of lower bounds on the value}
Our main tool to analyze a trilinear form that has a decomposition
with tight support
is the following theorem, which shows how to reduce the computation of a lower bound on its 
value to solving an optimization problem.
Its proof is given in the appendix.
\begin{theorem}\label{th:main}
Let $t$ be a trilinear form, and $D$ be a decomposition of $t$
with tight support $S\subseteq \Int\times \Int\times \Int$ and component
set $\{t(\myvec{s})\}_{\myvec{s}\in S}$.
Then, for any $P\in\D(S)$ and any $\rho\in[2,3]$,
\[
\log(V_\rho(t))\ge \sum_{\ell=1}^3\frac{H(P_\ell)}{3}+\sum_{\myvec{s}\in S} P(\myvec{s})\log (V_\rho(t(\myvec{s})))-\Gamma_S(P).
\]
\end{theorem}
This theorem can be seen as a generalized statement of the approach developed by Coppersmith
and Winograd~\cite{Coppersmith+90}.
Several similar statements already appeared in the literature. 
A weaker statement, corresponding to the simpler case where each component
is isomorphic to a matrix product (which removes the need for the term $-\Gamma_S(P)$ in the lower bound), 
can be found in \cite{Burgisser+97}. 
The generalization to the case of arbitrary components stated in Theorem~\ref{th:main}
was considered in~\cite{Stothers10}, and proved implicitly, 
by considering several cases (symmetric and asymmetric supports) and 
without reference to the entropy, 
in \cite{Davie+13,LeGallFOCS12,WilliamsSTOC12}.
Theorem~\ref{th:main} aims at providing a concise statement unifying all these results, 
described in terms of entropy in order to discuss the convexity of the lower bounds obtained.

\subsection{Solving the optimization problem}\label{sub:solving}
Let $t$ be a trilinear form with a decomposition that has a tight support, 
as in the statement of Theorem~\ref{th:main}.
It will be convenient to define, for any $\rho\in[2,3]$, 
the function $\Psi_{t,\rho}\colon \D(S)\to\Real$ as
\[
\Psi_{t,\rho}(P)=
\sum_{\ell=1}^3\frac{H(P_\ell)}{3}+\sum_{\myvec{s}\in S} P(\myvec{s})\log (V_\rho(t(\myvec{s})))
\]
for any $P\in \D(S)$. Note that this is a concave function on the convex set $\D(S)$.
In order to optimize the upper bound on $\log (V_\rho(t))$ that is obtained from Theorem \ref{th:main},
we would like to find, for a 
given value of~$\rho$, a probability distribution $P\in \D(S)$ that 
minimizes the expression
\[
\Gamma_S(P)-\Psi_{t,\rho}(P).
\]
This optimization problem is in general not convex, due to the presence of the term $\Gamma_S(P)$.
In this subsection we develop a method to overcome this difficulty and find, using Theorem \ref{th:main}, a lower bound on $V_\rho(t)$
in polynomial time.

Remember that 
$\Gamma_S(P)=
\max_{
Q
}
[H(Q)]
-H(P),
$
where the maximum is over all $Q\in \D(S)$ that are compatible with~$P$.
When $P$ is fixed, these conditions on $Q$ can be written as linear constraints.
Since the entropy is a strictly concave function, 
computing $-\Gamma_S(P)$ is then a strictly convex optimization problem on a convex set, and in particular has a unique solution $\hat Q$.
Note that 
$
\Gamma_S(\hat Q)=
H(\hat Q)
-H(\hat Q)=0,
$
and thus
$\Psi_{t,\rho}(\hat Q)$
is a lower bound on $\log (V_\rho(t))$. The tightness of this lower bound of course depends on the initial choice 
of~$P$. A natural choice is to take a probability distribution~$P$ that maximizes $\Psi_{t,\rho}(P)$, since finding 
such a probability distribution corresponds to solving a convex optimization problem.
This motivates the algorithm described in Figure~\ref{fig:algorithm}, which we call Algorithm $\mathcal{A}$.

\begin{figure}[ht!]
\begin{center}
\fbox{
\begin{minipage}{11.7 cm}
\underline{\textbf{Algorithm $\mathcal{A}$}}\vspace{2mm}

Input: $\bullet$ the support $S\subseteq \Int\times\Int\times\Int$ of the tensor $t$

$\phantom{Input: }$$\bullet$ a value $\rho\in[2,3]$

$\phantom{Input: }$$\bullet$ the values $V_\rho(t(\myvec{s}))$ for each $\myvec{s}\in S$\vspace{2mm}

1. Solve the following convex optimization problem.
\begin{equation*}
\left.
\begin{aligned}
&\text{minimize } -\Psi_{t,\rho}(P)\hspace{17mm}\\
&\text{subject to } P\in \D(S)
\end{aligned}
\hspace{1mm}\right\}
\text{OPT1}
\end{equation*}
\vspace{1mm}

2. Solve the following convex optimization problem, where $\hat P$ denotes  

\hspace{3mm} 
the solution found at Step 1.\vspace{2mm}
\begin{equation*}
\left.
\begin{aligned}
\hspace{0mm}&\text{minimize } -H(Q)\\
\hspace{0mm}&\text{subject to } Q\in \D(S)\\
\hspace{0mm}& \hspace{15mm}Q \text{ compatible with }\hat P
\end{aligned}
\hspace{1mm}\right\}
\text{OPT2}
\end{equation*}
\vspace{1mm}

3. Output $\Psi_{t,\rho}(\hat Q)$,
where $\hat Q$ denotes the solution 
found at Step 2.
\end{minipage}
}
\end{center}\vspace{-4mm}
\caption{Algorithm $\mathcal{A}$ computing, given a tensor $t$ with a decomposition that has a tight support and a value $\rho\in[2,3]$, 
a lower bound on $\log(V_\rho(t))$.}\label{fig:algorithm}
\end{figure}

As already mentioned, the optimization problem OPT 2 has a unique solution.
While the solution of the optimization problem OPT1 may not be unique,
it can actually be shown, using the strict concavity of the entropy function, that two solutions 
of OPT1 must have the same marginal probability distributions. 
Since the domain of the optimization problem OPT2 depends only on the marginal 
distributions of $\hat P$, the output of Algorithm $\mathcal{A}$ does not depend on which solution~$\hat P$ 
was found at Step 1. 
This output is thus unique and,
from Theorem \ref{th:main} and the discussion above,
it gives a lower bound on $\log(V_\rho(t))$. We state 
this conclusion 
in the following theorem.
\begin{theorem}\label{th:alg}
If the support of $t$ is tight, then
Algorithm~$\mathcal{A}$ outputs a lower bound on $\log(V_\rho(t))$. 
\end{theorem}

Let us now discuss the time complexity of implementing the algorithm of
Figure \ref{fig:algorithm}.
The worst-case running time depends on the time needed to solve 
the two optimization problems OPT1 and OPT2 at Steps 1 and 2. 
Let 
$v=\Psi_{t,\rho}(\hat Q)$
denote the output of an exact implementation of Algorithm~$\mathcal{A}$.
Theorem \ref{th:alg} shows that $v\le \log(V_\rho(t))$.
Since both OPT1 and OPT2 
are convex, and since the number of variables 
is upper bounded by $|S|$, for any $\varepsilon>0$
both problems can be solved with accuracy $\varepsilon$ in time 
$\poly(|S|,\log (1/\varepsilon))$ using standard methods~\cite{Ben-Tal+01,Fang+01}. 
Thus, for any $\varepsilon'>0$, we can compute in time 
$\poly(|S|,\log (1/\varepsilon'))$ a value $v'$ such that 
$
|v-v'|\le \varepsilon'\cdot v.
$
In particular, we can use $\frac{v'}{1+\varepsilon'}$ as a lower bound on $\log(V_\rho(t))$.

We finally explain how to exploit symmetries of the decomposition of $t$
to reduce the number of variables in Algorithm~$\mathcal{A}$.
These observations will enable us to slightly simplify the exposition of our results 
in the next sections. We first define invariance of a decomposition of a tensor.
\begin{definition}
Let $t$ be a tensor that has a decomposition $D$
with support $S$ and components $\{t(\myvec{s})\}_{\myvec{s}\in S}$.
The decomposition $D$ is $G$-invariant if 
$\Psi_{t,\rho}(P^\sigma)=\Psi_{t,\rho}(P)$ for any $P\in \D(S)$ and any $\sigma\in G$.
\end{definition}
With a slight abuse of language we will say, given a subgroup~$L$ of $\mathbb{S}_3$, 
that $D$ is $L$-invariant if $S$ is $L$-symmetric and $D$ is $L_S$-invariant (see Section \ref{sec:prelim}
for the definition of $L_S$).

Assume that the decomposition $D$ of the 
tensor $t$ on which we want to apply Algorithm $\mathcal{A}$ 
is $G$-invariant, where $G$ is a subgroup of $\mathbb{S}_{S}$. 
Consider the optimization problem OPT1. 
Since the value of its objective function is then
unchanged under the action of any permutation $\sigma\in G$ on $P$, 
OPT1 has a solution that is $G$-invariant (see, e.g., \cite{Boyd+04} for a discussion of symmetries in convex optimization).
Now, if $\hat P$ is $G$-invariant, then the (unique) solution of the optimization 
problem OPT2 is $G$-invariant as well, since the value of the function $-H(Q)$ is
unchanged under the action of any permutation on~$Q$.  
This means that, if the decomposition~$D$ is $G$-invariant, then $\D(S)$ can be replaced by $\D(S,G)$ at both
Steps~1 and~2 of Algorithm~$\mathcal{A}$.
Note that this set of distributions can be parametrized by 
$\dim(\F(S,G))$ parameters,
instead of $|S|$ parameters.

\subsection{Another approach}

In this subsection we describe another approach to obtain lower
bounds on $V_\rho(t)$ using Theorem \ref{th:main}, which is  
essentially how the powers of the construction
by Coppersmith and Winograd were studied in previous works
\cite{Davie+13,LeGallFOCS12,Stothers10,WilliamsSTOC12}.
Given any subgroup $G$ of $\mathbb{S}_S$, let us consider the
vector space $\F_0(S,G)$ of dimension $\chi(S,G)$ defined in Section \ref{sec:prelim}.
It will be convenient to represent functions in this vector space by vectors in
$\Real^{|S|}$, by fixing an arbitrary ordering of the elements in $S$.
Let $\myvec{R}$ be a generating matrix of size $|S|\times \chi(S,G)$ for $\F_0(S,G)$
(i.e., the columns of~$\myvec{R}$ form a basis of $\F_0(S,G)$).
Since each coordinate of $\Real^{|S|}$ corresponds to an element of~$S$, we 
write $R_{\myvec{s}j}$, for $\myvec{s}\in S$ and $j\in\{1,\ldots,\chi(S,G)\}$, 
to represent the element in the $\myvec{s}$-th
row and the $j$-th column of $\myvec{R}$.
The approach is based on the following proposition, 
which is similar to 
a characterization given in~\cite{WilliamsSTOC12}.

\begin{proposition}\label{lemma:kernel}
For any $P,P'\in \D(S,G)$ that are compatible,
the equality $\Gamma_S(P')=H(P)-H(P')$ holds if~$P$ satisfies the following two conditions:
\begin{itemize}
\item[(i)]
$P(\myvec{s})>0$ for any $s\in S$ such that~$\myvec{R}$ contains at least one non-zero entry
in its row labeled by $\myvec{s}$,
\item[(ii)]
$\sum_{\myvec{s}\in S}R_{\myvec{s}j}\log(P(\myvec{s}))=0$ for all  $j\in\{1,\ldots,\chi(S,G)\}$.
\end{itemize}
\end{proposition}
\begin{proof}
We first make the following observation:
for any vector $\myvec{x}\in \F_0(S,G)\subseteq \Real^{|S|}$, 
the equality
\[
\sum_{\myvec{s}\in S} x_{\myvec{s}}=0
\]
holds.
This implies in particular that $\sum_{\myvec{s}\in S}R_{\myvec{s}j}=0$ for any $j\in\{1,\ldots,\chi(S,G)\}$.

Let $\myvec{P'}$ denote the vector in $\Real^{|S|}$ representing the 
probability distribution $P'$. We have
\begin{align*}
\Gamma_S(P')&=
\max_{\myvec{u}\in C}
[H(\myvec{P'}+\myvec{R}\myvec{u})]
-H(\myvec{P'}),
\end{align*}
where
$
C=
\{
\myvec{u}\in \Real^{\chi(S,G)} \:|\: \myvec{P'}+\myvec{R}\myvec{u}\in \D(S)
\}.
$
Note that $C$ is a convex set.
Consider the function $h\colon C\to \Real$ defined as 
$h(\myvec{u})=H(\myvec{P'}+\myvec{R}\myvec{u})$ for any $\myvec{u}\in C$.
This is a concave function, differentiable on the interior of $C$. 
Let us take an interior point $\myvec{u}$ and write  $\myvec{z}=\myvec{P'}+\myvec{R}\myvec{u}$.
The partial derivatives at $\myvec{u}$ are
\[
\frac{\partial h}{\partial u_j}=-\sum_{\myvec{s}\in S}R_{\myvec{s}j}(1+\log(z_\myvec{s}))=-\sum_{\myvec{s}\in S}R_{\myvec{s}j}\log(z_\myvec{s}),
\]
for each $j\in\{1,\ldots,\chi(S,G)\}$.

A probability distribution $P$ that satisfies the conditions in the statement of the proposition
therefore corresponds to a vanishing point (and thus a global maximum) of the function~$h$, which implies that
the equality
$\Gamma_S(P')=H(P)-H(P')$ holds for this distribution $P$, as claimed. 
\end{proof}

In particular, applying Proposition \ref{lemma:kernel} with $P'=P$ shows that, if Conditions (i) and~(ii) are satisfied, 
then $\Gamma_S(P)=0$, which implies $\log(V_\rho(t))\ge \Psi_{t,\rho}(P)$ from Theorem \ref{th:main}.
This motivates the algorithm described in Figure \ref{fig:algorithmB} that outputs a lower bound on $\log(V_\rho(t))$.
We will call it Algorithm $\mathcal{B}$. 

\begin{figure}[ht!]
\begin{center}
\fbox{
\begin{minipage}{11.7 cm} 
\underline{\textbf{Algorithm $\mathcal{B}$}}\vspace{2mm}

Input: $\bullet$ the support $S\subseteq \Int\times\Int\times\Int$ for the tensor $t$

$\phantom{Input: }$$\bullet$ a value $\rho\in[2,3]$

$\phantom{Input: }$$\bullet$ the values $V_\rho(t(\myvec{s}))$ for each $\myvec{s}\in S$

$\phantom{Input: }$$\bullet$ a subgroup $G$ of $\mathbb{S}_S$ such that the 
decomposition of $t$ is $G$-invariant 
\vspace{2mm}

1. Solve the following optimization problem.
\begin{equation*}
\begin{aligned}
\hspace{3mm}&\text{minimize } -\Psi_{t,\rho}(P)\hspace{17mm}\\
\hspace{3mm}&\text{subject to } P\in \D(S,G)\\
\hspace{3mm}&\phantom{subject to } P \text{ satisfies Conditions~(i)-(ii) of Proposition }\ref{lemma:kernel}
\end{aligned}
\hspace{1mm}
\end{equation*}
\vspace{1mm}

2. Output $\Psi_{t,\rho}(\tilde P)$,
where $\tilde P$ denotes the solution 
found at Step 1.
\end{minipage}
}
\end{center}\vspace{-4mm}
\caption{Algorithm $\mathcal{B}$ computing, given a tensor $t$ with a decomposition that has a tight support and a value $\rho\in[2,3]$, 
a lower bound on $\log(V_\rho(t))$.}\label{fig:algorithmB}
\end{figure}

Note that, when $\chi(S,G)=0$, Algorithms $\mathcal{A}$ and $\mathcal{B}$ solve
exactly the same optimization problem 
(since in Algorithm~$\mathcal{B}$ Conditions (i) and (ii) are satisfied for any $P\in D(S,G)$,
and $\hat Q=\hat P$ in Algorithm~$\mathcal{A}$) and thus output the same value.
When $\chi(S,G)>0$ Algorithm $\mathcal{B}$ usually gives better lower bounds than Algorithm $\mathcal{A}$,
but at the price of introducing $\chi(S,G)$ highly nonconvex constraints, which makes the optimization problem 
much harder to solve, both in theory and in practice, even for a modest number of variables.

\section{Powers of Tensors}\label{section:powers}
Let $t$ and $t'$ be two trilinear forms with decompositions $D$ and $D'$, respectively. 
Let $\supp(t)\subset \Int\times\Int\times\Int$ 
and $\supp(t')\subset \Int\times\Int\times\Int$ denote their supports, 
and $\{t(\myvec{s})\}_{\myvec{s}\in\supp(t)}$ and $\{t'(\myvec{s'})\}_{\myvec{s'}\in\supp(t')}$ denote their component sets.
Assume that both supports are tight.
Fix $\rho\in[2,3]$ and assume that lower bounds on the values $V_\rho(t(\myvec{s}))$ and $V_\rho(t'(\myvec{s'}))$ are
known for each $\myvec{s}\in \supp(t)$ and each $\myvec{s'}\in \supp(t')$.
In this section we describe a method, inspired by \cite{Coppersmith+90} and \cite{Stothers10},
and also used in \cite{WilliamsSTOC12},
to analyze $V_\rho(t\otimes t')$, and then show how to use it
to analyze
$V_\rho(t^{\otimes m})$ when $m$ is a 
power of two.

In this section we will denote $\alpha_1,\alpha_2,\alpha_3\colon \Int\times\Int\times\Int\to \Int$ the 
three coordinate functions of $\Int\times\Int\times\Int$.

Consider the tensor 
\[
t \otimes t'=\sum_{\myvec{s}\in \supp(t)}\sum_{\myvec{s'}\in \supp(t')} t(\myvec{s})\otimes t'(\myvec{s'}).
\]
Consider the following decomposition of $t\otimes t'$: the support is
\[
\supp(t\otimes t')=
\left\{
(
\alpha_1(\myvec{s})+\alpha_1(\myvec{s'}),
\alpha_2(\myvec{s})+\alpha_2(\myvec{s'}),
\alpha_3(\myvec{s})+\alpha_3(\myvec{s'})
)
\:|\:
\myvec{s}\in \supp(t),\myvec{s'}\in \supp(t')
\right\}
\]
and,
for each $(a,b,c)\in \supp(t\otimes t')$
the associated component is 
\[
(t\otimes t')(a,b,c)=\sum t(\myvec{s})\otimes t'(\myvec{s'}),
\]
where the sum is over all 
$(\myvec{s},\myvec{s'})\in \supp(t)\times \supp(t')$
such that $\alpha_1(\myvec{s})+\alpha_1(\myvec{s'})=a$,
$\alpha_2(\myvec{s})+\alpha_2(\myvec{s'})=b$ and
$\alpha_3(\myvec{s})+\alpha_3(\myvec{s'})=c$.
Note that the support of this decomposition is tight.
If
lower bounds on the value of each component 
are known,
then we can use this decomposition to obtain 
a lower bound on $V_\rho(t\otimes t')$, by using
Algorithm $\mathcal{A}$ on $t\otimes t'$,
which requires solving 
two convex optimization problems,
 each having $|\supp(t\otimes t')|$ variables.

We now explain how to evaluate the value of those
components $(t\otimes t')(a,b,c)$.
For any $(a,b,c)\in \supp(t\otimes t')$, consider the following
decomposition of $(t\otimes t')(a,b,c)$:
the support is 
\[
\left\{
\myvec{s}\in\supp(t)\:|\:(a-\alpha_1(\myvec{s}),b-\alpha_2(\myvec{s}),c-\alpha_3(\myvec{s}))\in\supp(t')
\right\}
\]
and, for each element $\myvec{s}$ in this set, the corresponding component is 
\[
t(\myvec{s})\otimes t'(a-\alpha_1(\myvec{s}),b-\alpha_2(\myvec{s}),c-\alpha_3(\myvec{s})).
\]
Note that the support in this decomposition is tight, and has size at most $|\supp(t)|$.
The value of each component 
can be lower bounded as
\begin{align*}
V_\rho&\left( t(\myvec{s})\otimes t'(a-\alpha_1(\myvec{s}),b-\alpha_2(\myvec{s}),c-\alpha_3(\myvec{s}))\right)\ge\\
&\hspace{4mm}V_\rho(t(\myvec{s}))\times V_\rho(t'(a-\alpha_1(\myvec{s}),b-\alpha_2(\myvec{s}),c-\alpha_3(\myvec{s})),
\end{align*}
from the supermultiplicativity of the value. As we supposed that the lower bounds on the values of each component of $t$ and~$t'$ are known, we can use
Algorithm $\mathcal{A}$ on each $(t\otimes t')(a,b,c)$ 
to obtain a lower bound on $V_\rho((t\otimes t')(a,b,c))$, which requires solving 
two convex optimization problems, each having at most $|\supp(t)|$ variables.

Let us now consider the case $t'=t$.
We have just shown the following result:
a lower bound on $V_\rho(t^{\otimes 2})$
can be computed by solving two convex optimization problems with
$|\supp(t^{\otimes 2})|$ variables, 
and $2|\supp(t^{\otimes 2})|$ convex optimization
problems with at most $|\supp(t)|$ variables.  
An important point is that this method additionally gives,
as described in the previous paragraphs,  
a decomposition of $t^{\otimes 2}$ with tight support, and
a lower bound on $V_\rho(t^{\otimes 2}(a,b,c))$ 
for each component $t^{\otimes 2}(a,b,c)$.
This information
can then be used to analyze the trilinear form
$t^{\otimes 4}=t^{\otimes 2}\otimes t^{\otimes 2}$, 
by replacing $t$ by $t^{\otimes 2}$ in the above analysis,
giving a decomposition 
of $t^{\otimes 4}$ with tight support, a lower bound on $V_\rho(t^{\otimes 4})$
and a lower bound on the value $V_\rho(t^{\otimes 4}(a,b,c))$ of each component.
By iterating this approach $r$ times, for any $r\ge 1$,
we can analyze the trilinear form $t^{\otimes 2^r}$, 
and in particular obtain
a lower bound on $V_\rho(t^{\otimes 2^r})$. 
Let us denote by~$D^{2^r}$
the decomposition of $t^{\otimes 2^r}$ obtained by this approach.
Its support is
\[
\supp(t^{\otimes 2^{r}})=
\left\{
(
\alpha_1(\myvec{s})+\alpha_1(\myvec{s'}),
\alpha_2(\myvec{s})+\alpha_2(\myvec{s'}),
\alpha_3(\myvec{s})+\alpha_3(\myvec{s'})
)\:|\:
\myvec{s},\myvec{s'}\in \supp(t^{\otimes 2^{r-1}})
\right\}
\]
and,
for any $(a,b,c)\in \supp(t^{\otimes 2^r})$, 
the corresponding component is
\[
t^{\otimes 2^r}(a,b,c)=\sum t^{\otimes 2^{r-1}}(\myvec{s})\otimes t^{\otimes 2^{r-1}}(\myvec{s'}),
\]
where the sum is over all 
$(\myvec{s},\myvec{s'})\in \supp(t^{\otimes 2^{r-1}})$
such that $\alpha_1(\myvec{s})+\alpha_1(\myvec{s'})=a$,
$\alpha_2(\myvec{s})+\alpha_2(\myvec{s'})=b$ and
$\alpha_3(\myvec{s})+\alpha_3(\myvec{s'})=c$.
This approach also gives a decomposition 
$D_{abc}^{2^r}$  
of each component 
$t^{\otimes 2^r}(a,b,c)$. 
In this decomposition the support, which we denote $S^{2^r}_{abc}$, is
\[
S^{2^r}_{abc}=
\left\{
\myvec{s}\in\supp(t^{\otimes 2^{r-1}})
\:|\:
(a-\alpha_1(\myvec{s}),b-\alpha_2(\myvec{s}),c-\alpha_3(\myvec{s}))\in\supp(t^{\otimes 2^{r-1}})
\right\}
\]
and, for any $\myvec{s}\in S^{2^r}_{abc}$, the corresponding component of $t^{\otimes 2^r}(a,b,c)$ is 
\[
t^{\otimes 2^{r-1}}(\myvec{s})\otimes t^{\otimes 2^{r-1}}(a-\alpha_1(\myvec{s}),b-\alpha_2(\myvec{s}),c-\alpha_3(\myvec{s})).
\]

The 
overall number of convex optimization problems that need to be solved in order 
to analyze $t^{\otimes 2^r}$ by the above approach is upper bounded by
$r(2+2|\supp(t^{\otimes 2^{r}})|)$, while the number of variables in each optimization problem is 
upper bounded by $|\supp(t^{\otimes 2^{r}})|$.
In the case where $\supp(t)$ is $b$-tight we can give a simple upper bound on this quantity.
Indeed, when $\supp(t)$ is $b$-tight, our construction 
guarantees that 
$\supp(t^{\otimes 2^{r}})$ is $(b2^r)$-tight, which implies that 
$
|\supp(t^{\otimes 2^{r}})|\le (b2^r)^2.
$
We thus obtain the following result.
\begin{theorem}
Let $t$ be a trilinear form that 
has a decomposition with $b$-tight support $\supp(t)$
and components $\{t(\myvec{s})\}$.
Fix $\rho\in[2,3]$ and assume that a lower bound on the value $V_\rho(t(\myvec{s}))$ is
known for each $\myvec{s}\in\supp(t)$. 
Then, for any integer $r\ge 1$, 
a lower bound on $V_\rho(t^{\otimes 2^r})$
can be computed by solving $\poly(b,2^r)$ 
convex optimizations problems, each
optimization problem having
$\poly(b,2^r)$ variables. 
\end{theorem}

Finally, we present two simple lemmas that show how to exploit the symmetries of the decomposition of $t$ to 
slightly reduce the number of variables in our optimization problems. 
\begin{lemma}\label{lemma:sym3}
For any $r\ge 1$ and any $(a,b,c)\in \supp(t^{\otimes 2^r})$,
the decomposition $D_{abc}^{2^r}$ is
$\{\identi,\pi\}$-invariant, where $\identi$ denotes the identity permutation 
and $\pi$ is the permutation on $S^{2^r}_{abc}$
such that 
\[
\pi(\myvec{s})=(a-\alpha_1(\myvec{s}),b-\alpha_2(\myvec{s}),c-\alpha_3(\myvec{s}))
\]
for all $\myvec{s}\in S^{2^r}_{abc}$.
\end{lemma}
\begin{proof}
Observe that, for any probability distribution $P\in D(S_{abc}^{2^r})$, the equality
$H(P_\ell)=H(P^\pi_\ell)$ holds for all $\ell\in\{1,2,3\}$, 
 which implies that 
 \[
 \Psi_{t^{\otimes 2^r}(a,b,c),\rho}(P)=\Psi_{t^{\otimes 2^r}(a,b,c),\rho}(P^\pi),
 \]
 as wanted.
\end{proof}
\begin{lemma}\label{lemma:sym2}
Let $L$ be a subgroup of $\mathbb{S}_3$.
Assume that $\supp(t)$ is $L$-symmetric and that
\[
V_\rho(t(s_1,s_2,s_3))=V_\rho(t(s_{\sigma(1)},s_{\sigma(2)},s_{\sigma(3)}))
\]
for any $\sigma\in L$ and any $\myvec{s}=
(s_1,s_2,s_3)\in \supp(t)$.
Then $D$ is $L$-invariant and, for any $r\ge 1$,
the decomposition $D^{2^r}$ is $L$-invariant as well. 
\end{lemma}
\begin{proof}
For any $\sigma\in L$, any probability distribution $P\in\D(\supp(t))$ and any $\ell\in\{1,2,3\}$,
the equality $P_\ell=P_{\sigma(\ell)}^{\pi_\sigma}$ holds,
where $\pi_\sigma$ denotes the permutation 
such that 
\[
\pi_\sigma(s_1,s_2,s_3)=(s_{\sigma(1)},s_{\sigma(2)},s_{\sigma(3)})
\] 
for 
all $(s_1,s_2,s_3)\in \supp(t)$. 
This implies that 
$\Psi_{t,\rho}(P)=\Psi_{t,\rho}(P^{\pi_\sigma})$, and thus $D$ is $L$-invariant.
The same argument shows that $D^{2^r}$ is $L$-invariant for any $r\ge 1$. 
\end{proof}

From the discussion of Section \ref{sub:solving},
Lemma \ref{lemma:sym3} enables us to reduce the number of variables when 
computing the lower bound on 
$V_\rho(t^{\otimes 2^r}(a,b,c))$ using Algorithm $\mathcal{A}$:
instead of solving an optimization problem over 
$\D(\supp(t^{\otimes 2^r}(a,b,c)))$,
we only need to consider  
$\D(\supp(t^{\otimes 2^r}(a,b,c)),\{\identi,\pi\})$.
Similarly, if the conditions of Lemma \ref{lemma:sym2} are satisfied, 
then, instead of considering 
$\D(\supp(t^{\otimes 2^r}))$,
we need only to consider 
$\D(\supp(t^{\otimes 2^r}),L)$
when computing the lower bound on $V_\rho(t^{\otimes 2^r})$ using Algorithm $\mathcal{A}$.
\vspace{2mm}

\noindent\textbf{Remark.}
The approach described in this section can be generalized 
to obtain lower bounds on $V_\rho(t^{\otimes m})$ 
when $m$ is not a power of two. For instance the third power
can be analyzed by studying $t\otimes t'$ with $t'=t^{\otimes 2}$.
Another possible straightforward generalization is to allow other linear dependences in the definition of the support, i.e., 
defining the support of $t\otimes t'$ as
\[
\supp(t\otimes t')=
\left\{
(
\alpha_1(\myvec{s})+u\alpha_1(\myvec{s'}),
\alpha_2(\myvec{s})+u\alpha_2(\myvec{s'}),
\alpha_3(\myvec{s})+u\alpha_3(\myvec{s'})
)
\:|\:
\myvec{s}\in \supp(t),\myvec{s'}\in \supp(t')
\right\}
\]
where $u\in \Int$ can be freely chosen.
These two generalizations nevertheless 
do not seem to lead to any improvement for $\omega$ when applied
to existing constructions.

\section{Application}
In this section we apply the theory
developed in the previous sections 
to the construction $\A$ by Coppersmith and Winograd, 
in order to obtain upper 
bounds on $\omega$.

\subsection{Construction}
Let $\field$ be an arbitrary field.
Let $q$ be a positive integer, and consider three 
vector spaces $U$, $V$ and $W$ of dimension $q+2$ over $\field$. 
Take 
a basis $\{x_0,\ldots,x_{q+1}\}$ of $U$, 
a basis $\{y_0,\ldots,y_{q+1}\}$ of $V$,
and
a basis $\{z_0,\ldots,z_{q+1}\}$ of $W$. 

The trilinear form $\A$ considered by Coppersmith 
and Winograd is the following trilinear form on
$(U,V,W)$:
\begin{align*}
\A=\sum_{i=1}^q (x_0y_iz_i+x_iy_0z_i+&x_iy_iz_0)+x_0y_0z_{q+1}+
x_0y_{q+1}z_{0}+x_{q+1}y_0z_0.
\end{align*}
It was shown in \cite{Coppersmith+90} 
that $\underline R(\A)=q+2$.
Consider the following decomposition of $U$, $V$ and $W$:
\[
U=U_0\oplus U_1\oplus U_2,\:\:\: V=V_0\oplus V_1\oplus V_2,\:\:\: W=W_0\oplus W_1\oplus W_2,
\]
where
$U_0=\myspan{x_0}$,
$U_1=\myspan{x_1,\ldots,x_q}$ and $U_2=\myspan{x_{q+1}}$,
$V_0=\myspan{y_0}$, $V_1=\myspan{y_1,\ldots,y_{q}}$ and $V_2=\myspan{y_{q+1}}$,
$W_0=\myspan{z_0}$, $W_1=\myspan{z_1,\ldots,z_{q}}$ and $W_2=\myspan{z_{q+1}}$.
This decomposition induces a decomposition $D$ of $\A$ with tight support
\[
S=\{(2,0,0),(1,1,0),(1,0,1),(0,2,0),(0,1,1),(0,0,2)\}.
\]
The components associated with $(2,0,0)$ and $(1,1,0)$ are
\begin{align*}
\A(2,0,0)&=x_{q+1}y_0z_0\cong \braket{1,1,1},\\
\A(1,1,0)&=\sum_{i=1}^qx_iy_iz_0\cong \braket{1,q,1}.
\end{align*}
We have $V_\rho(\A(2,0,0))=1$ and $V_\rho(\A(1,1,0))\ge q^{\rho/3}$,
from the definition of the value.
The other components $\A(0,2,0)$ and $\A(0,0,2)$ are obtained 
by permuting the coordinates of $\A(2,0,0)$,
while the components $\A(1,0,1)$ and $\A(0,1,1)$ are obtained 
by permuting the coordinates
of $\A(1,1,0)$.

We now use Theorem \ref{th:main} to obtain an upper bound on~$\omega$.
Let $P$ be a probability distribution in $\D(S)$. Let us write
$P(2,0,0)=a_1$, $P(1,1,0)=a_2$, $P(1,0,1)=a_3$, 
$P(0,2,0)=a_4$, $P(0,1,1)=a_5$ and $P(0,0,2)=a_6$. 
The marginal distributions of $P$ are
$\myvec{P_1}=(a_1,a_2+a_3,a_4+a_5+a_6)$, 
$\myvec{P_2}=(a_4,a_2+a_5,a_1+a_3+a_6)$
and $\myvec{P_3}=(a_6,a_3+a_5,a_1+a_2+a_4)$.
Since the only element in $D(S)$ compatible with~$P$ is 
$P$, we have $\Gamma_S(P)=0$. 
Theorem~\ref{th:main} 
thus implies that  
\[
V_\rho(\A)\ge\exp\left(\frac{H(P_1)+H(P_2)+H(P_3)}{3}\right)\times q^{(a_2+a_3+a_5)\rho/3}
\]
for any $\rho\in[2,3]$.
Evaluating this expression with 
$q=6$, $a_2=a_3=a_5=0.3173$, $a_1=a_4=a_6=(1-3a_2)/3$,
and $\rho=2.38719$ gives $V_\rho(\A)> 8.00000017$.
Using Theorem \ref{th:value} and the fact that $\underline R(\A)=q+2$, 
we conclude that $\omega<2.38719$. This is the same upper bound as the bound
obtained in Section 7 of~\cite{Coppersmith+90}.

\subsection{Analyzing the powers using Algorithm $\boldsymbol{\mathcal{A}}$}
For any $r\ge 1$, we now consider the tensor $\A^{\otimes 2^r}$ and analyze 
it using the framework and the notations of Section \ref{section:powers}.
The support of its decomposition $D^{2^r}$ is the set of all triples
\[
(a,b,c)\in\{0,\ldots,2^{r+1}\}\times \{0,\ldots,2^{r+1}\}\times\{0,\ldots,2^{r+1}\}
\]
such that $a+b+c=2^{r+1}$.
Note that the decomposition~$D$ of~$\A$ satisfies the conditions of Lemma~\ref{lemma:sym2}
for the subgroup $L=\mathbb{S}_3$ of $\mathbb{S}_3$, which implies that 
$D^{2^r}$ is $\mathbb{S}_3$-invariant. Thus, from the discussion in Section~\ref{sub:solving},
when applying Algorithm~$\mathcal{A}$ on the trilinear form
$t^{\otimes 2^r}$ in order to obtain a lower bound on $V_\rho(t^{2^r})$,
we only need to consider probability 
distributions in $\D(\supp(\A^{\otimes 2^r}),\mathbb{S}_3)$.
This set 
can be parametrized by 
$
\dim(\F(\supp(\A^{\otimes 2^r}),\mathbb{S}_3))
$
parameters. Remember that we also need 
a lower bound on the value of each component $\A^{\otimes 2^r}(a,b,c)$
before applying $\mathcal{A}$ on~$t^{\otimes 2^r}$. Using the method
described in Section~\ref{section:powers}, these lower bounds 
are computed recursively 
by applying Algorithm~$\mathcal{A}$ on the decomposition 
$D^{2^r}_{abc}$ of the component. 
Actually, we do not need to apply~$\mathcal{A}$ when $a=0$, $b=0$ or $c=0$,
since a lower bound on the value can be found analytically in this case,
as stated in the following lemma (see Claim 7 in~\cite{WilliamsFull} for a proof).
\begin{lemma}\label{lemma:value-B}
For any $r\ge 0$ and any $b\in\{0,1,\ldots,2^r\}$, 
\[
V_\rho\big(\A^{\otimes 2^r}\!(2^{r+1}-b,b,0)\big)\ge 
\!
\left(
\sum_{\begin{subarray}{c}e\in\{0,\ldots,b\}\\e\equiv b\bmod 2\end{subarray}}
\frac{2^r!}{e!(\frac{b-e}{2})!(2^r-\frac{b+e}{2})!}q^e
\right)^{\!\!\!\frac{\rho}{3}}\!\!\!.
\]
\end{lemma}
Table \ref{table:B-param} presents, for $r\in\{1,2,3,4,5\}$,
the number of variables in the global optimization problem, 
the compatibility degree, and the best upper bound 
on $\omega$ we obtained by this approach. The programs used 
to derive these upper bounds can be found at \cite{files}, 
and use the Matlab software CVX for convex optimization.
We work out below in details the cases $r=1$ and $r=2$.\vspace{-3mm}

\begin{table}[h!]
\renewcommand\arraystretch{1}
\begin{center}
\caption{Analysis of $\A^{\otimes 2^r}$ using Algorithm $\mathcal{A}$.}\label{table:B-param}\vspace{1mm}
\begin{tabular}{|c|c|c|c|}
\hline
$\!r\!$&$\dim(\F(\supp(\A^{\otimes 2^r}),\mathbb{S}_3))$&$\chi(\supp(\A^{\otimes 2^r}),\mathbb{S}_3)$&upper bound obtained\bigstrut\\ 
\hline
\!1\!& 4&0&$\omega<2.3754770$\\
\!2\!& 10&2&$\omega<2.3729372$\\
\!3\!& 30&14&$\omega<2.3728675$\\
\!4\!& 102&70&$\omega<2.3728672$\\
\!5\!& 374&310&$\omega<2.3728671$\\
\hline
\end{tabular}
\end{center}
\end{table}\vspace{-3mm}

The case $r=1$ is easy to deal with, since the compatibility degree is zero,
and the 
values of all components but one can be computed directly using Lemma \ref{lemma:value-B}:
we have $V_\rho(\A^{\otimes 2}(4,0,0))=1$, $V_\rho(\A^{\otimes 2}(3,1,0))=(2q)^{\rho/3}$ and
$V_\rho(\A^{\otimes 2}(2,2,0))=(q^2+2)^{\rho/3}$.
Let $P$ be a probability distribution in $\D(\supp(t^{\otimes 2}),\mathbb{S}_3)$.
Write $P(4,0,0)=a_1$, $P(3,1,0)=a_2$, $P(2,2,0)=a_3$ and
$P(2,1,1)=a_4$. 
The partial distributions of $P$ are
\[
\myvec{P_1}=\myvec{P_2}=\myvec{P_3}=(a_1,2a_2,2a_3+a_4,2a_2+2a_4,2a_1+2a_2+a_3).
\]
Since $\chi(\supp(\A^{\otimes 2}),\mathbb{S}_3)=0$, 
we have $\Gamma_{\supp(t^{\otimes 2})}(P)=0$
and thus
Theorem \ref{th:main} 
gives the lower bound 
\[
V_\rho(\A^{\otimes 2})
\ge
\frac{(2q)^{6a_2\rho/3}(q^2+2)^{3a_3\rho/3}
\!\times [V_\rho(\A^{\otimes 2}(2,1,1))]^{3a_4}}
{a_1^{a_1}(2a_2)^{2a_2}(2a_3+a_4)^{2a_3+a_4}(2a_2+2a_4)^{2a_2+2a_4}(2a_1+2a_2+a_3)^{2a_1+2a_2+a_3}}.
\]
The value of the component $\A^{\otimes 2}(2,1,1)$ is computed as described in Section \ref{section:powers},
by considering its decomposition $D^2_{211}$. This decomposition corresponds to the identity
\begin{align*}
\A^{\otimes 2}(2,1,1)=
&\A(2,0,0)\otimes \A(0,1,1)+\A(1,1,0)\otimes \A(1,0,1)+\\
&\A(1,0,1)\otimes \A(1,1,0)+\A(0,1,1)\otimes \A(2,0,0),
\end{align*}
and has (tight) support $S^2_{211}\!=\!\{(2,0,0),(1,1,0),(1,0,1),(0,1,1)\}$. 
Lemma~\ref{lemma:sym3} shows that 
this decomposition is $\{\identi,\pi\}$-invariant,
where $\pi$ is the permutation that exchanges $(2,0,0)$ and $(0,1,1)$
and exchanges $(1,1,0)$ and $(1,0,1)$.
Let $P'$ be a distribution on $\D(S^2_{211},\{\identi,\pi\})$, and write $P'(2,0,0)=P'(0,1,1)=b_1$ and 
$P'(1,1,0)=P'(1,0,1)=b_2$. The partial distributions are 
$\myvec{P'_1}=(b_1,2b_2,b_1)$ and 
$\myvec{P'_2}=\myvec{P'_3}=(b_1+b_2,b_1+b_2)$.
Since $\Gamma_{S^2_{211}}(P')=0$, we have
\[
V_\rho(\A^{\otimes 2}(2,1,1))\ge\frac{(1\times q^{\rho/3})^{2b_1}(q^{\rho/3}\times q^{\rho/3})^{2b_2}}
{\left[b_1^{2b_1}(2b_2)^{2b_2}(b_1+b_2)^{4(b_1+b_2)}\right]^{1/3}}.
\]
For $\rho=2.3754770$ and $q=6$, we obtain $V_\rho(\A^{\otimes 2}(2,1,1))>27.35608$
by taking $b_1=0.01378$ and $b_2=0.48622$. Then, taking
$a_1=0.00023$, $a_2=0.01250$, $a_3=0.10254$ and $a_4=(1-3a_1-6a_2-3a_3)/3$,
we obtain 
\[
V_\rho(\A^{\otimes 2})>64.00000357>(q+2)^2.
\]
Using Theorem \ref{th:value} and the fact that $\underline R(\A^{\otimes 2})\le (q+2)^2$, 
we conclude that $\omega<2.3754770$. This is the same upper bound as the bound  found in Section 8 of
\cite{Coppersmith+90}.

For $r=2$, we first need to compute lower bounds on the values of the ten components $\A^{\otimes 4}(a,b,c)$. 
Five of them can be computed directly using Lemma \ref{lemma:value-B}, while the remaining five are computed
using Algorithm $\mathcal{A}$. Table \ref{table:2} gives the lower bounds obtained for $\rho=2.3729372$ and $q=5$. 
We then apply 
Algorithm $\mathcal{A}$ on $\A^{\otimes 4}$, and obtain 
\[
V_{\rho}(\A^{\otimes 4})>2401.00013>(q+2)^4
\]
for
the probability distributions $\hat P$ and~$\hat Q$ 
given in Table \ref{table:2}, which gives the upper bound
$\omega<2.3729372$.

\begin{table}[h!]
\renewcommand\arraystretch{1}
\begin{center}
\caption{
Values of the components and optimal probability distributions for $\A^{\otimes 4}$, 
with $\rho=2.3729372$ and $q=5$, computed using Algorithm $\mathcal{A}$. 
In this table, $d$ represents $\dim(\F(S^4_{abc},\{\identi,\pi\}))$ 
and $\chi$ represents $\chi(S^4_{abc},\{\identi,\pi\})$, where $\pi$
is defined in Lemma \ref{lemma:sym3}.
The symbol $-$ means that the value is not relevant, since lower bounds on the
values can be computed directly using Lemma \ref{lemma:value-B}.}\label{table:2}\vspace{3mm}
\begin{tabular}{|c|c|c|r|c|c|}
\hline
\multirow{2}{*}{$\!\!(a,b,c)\!\!$}& 
\multirow{2}{*}{$\!d\!$}& 
\multirow{2}{*}{$\!\chi\!$}& 
lower bound on&
\multirow{2}{*}{$\hat P(a,b,c)$}&
\multirow{2}{*}{$\hat Q(a,b,c)$}\\ 
&&&$V_\rho(\A^{\otimes 4}(a,b,c))$&&\\
\hline
(8,0,0)&\!--\!&\!--\!&1.000000&    0.00000013&0.00000013\\
(7,1,0)&\!--\!&\!--\!&10.692703&    0.00001649&0.00001649\\
(6,2,0)&\!--\!&\!--\!&53.738198&  0.00000000&0.00049685\\
(6,1,1)&\!2\!&\!0\!&66.354789&  0.00164314&0.00064945\\
(5,3,0)&\!--\!&\!--\!&149.196694&  0.01178178&0.00474178\\
(5,2,1)&\!3\!&\!0\!&235.605709&  0.00259744&0.00963744\\
(4,4,0)&\!--\!&\!--\!&223.037068& 0.00000000&0.01308631\\
(4,3,1)&\!4\!&\!0\!&472.727437&  0.05233355&0.04628725\\
(4,2,2)&\!5\!&\!0\!&605.359824&0.08605608&0.07197608\\
(3,3,2)&\!5\!&\!1\!&793.438218&0.11217546&0.12526177\\
\hline
\end{tabular}
\end{center}
\end{table}

\subsection{Analyzing the powers using both Algorithms $\boldsymbol{\mathcal{A}}$ and $\boldsymbol{\mathcal{B}}$}\label{sub:mixed}
As mentioned in the introduction, the best known upper bound on $\omega$
obtained from the fourth power of~$\A$ is 
$\omega<2.3729269$, which is slightly better than what
we obtained in the previous subsection using Algorithm
$\mathcal{A}$. This better bound can actually be obtained by using 
Algorithm~$\mathcal{B}$ instead of Algorithm $\mathcal{A}$
when computing the lower bound on $V_\rho(\A^{\otimes 4})$.
More precisely, in this case the optimization problem in 
Algorithm $\mathcal{B}$ 
asks to minimize $-\Psi_{t,\rho}(P)$ such that 
$P\in \D(\supp(\A^{\otimes 4}),\mathbb{S}_3)$ 
and $P$ satisfies two additional constraints, since 
$\chi(\supp(\A^{\otimes 4}),\mathbb{S}_3)=2$. These two constraints (the same as in \cite{Davie+13,Stothers10,WilliamsSTOC12}) are:
\begin{align*}
\log(P(6,2,0))+&2\log(P(4,3,1))-\log(P(4,4,0))\\
&-\log(P(6,1,1))-\log(P(3,3,2))=0,\\
\log(P(5,3,0))+&\log(P(4,3,1))+\log(P(4,2,2))\\
-\log(P(4,4,0))&-\log(P(5,2,1))-\log(P(3,3,2))=0.
\end{align*}
They are highly non-convex but, since their number is only two,
the resulting optimization problem can be solved fairly easily,
giving the same upper bound $\omega<2.3729269$ as the bound reported in \cite{WilliamsSTOC12}.

We can also use Algorithm $\mathcal{B}$ instead of Algorithm $\mathcal{A}$
to analyze $\A^{\otimes 8}$, but 
solving the corresponding optimization problems in this case was delicate and required a combination of
several tools.
We obtained lower bounds on the values of each component by solving 
the non-convex optimization problems using the NLPSolve function in Maple,
while the lower bound on $V_\rho(\A^{\otimes 8})$ has been obtained by solving 
the corresponding optimization problem (with 30 variables and 14 non-convex constraints)
using the fmincon function in Matlab. All the programs used are available at \cite{files}, and
the numerical solutions are given for 
$\rho=2.3728642$ and $q=5$ in Table \ref{table:B-power8}. 
The probability distribution of Table \ref{table:B-power8} gives
\[
V_\rho(\A^{\otimes 8})>5764802.8>(q+2)^8,
\]
which shows that
$\omega<2.3728642$.

\begin{table}[h!]
\renewcommand\arraystretch{1}
\begin{center}
\caption{Values of the components and optimal probability distribution for $\A^{\otimes 8}$, 
with $\rho=2.3728642$ and $q=5$, computed using Algorithm $\mathcal{B}$. 
In this table, $d$ represents $\dim(\F(S^{8}_{abc},\{\identi,\pi\}))$ 
and $\chi$ represents $\chi(S^8_{abc},\{\identi,\pi\})$, where $\pi$
is defined in Lemma \ref{lemma:sym3}.
The symbol $-$ means that the value is not relevant, since lower bounds on the
values can be computed directly using Lemma \ref{lemma:value-B}.}\label{table:B-power8}\vspace{3mm}
\begin{tabular}{|c|c|c|r|c|}
\hline
\multirow{2}{*}{$\!(a,b,c)\!$}& 
\multirow{2}{*}{$d$}& 
\multirow{2}{*}{$\chi$}& 
lower bound on&
\multirow{2}{*}{$\tilde P(a,b,c)$}\\ 
&&&$V_\rho(\A^{\otimes 8}(a,b,c))$&\\
\hline
(16,0,0)& --&--&1.0000&0.000000000049\\
(15,1,0)& --&--&18.4993&0.000000000045\\
(14,2,0)& --&--&179.5755&0.000000001353\\
(14,1,1)& 2&0&202.5694&0.000000001461\\
(13,3,0)& --&--&1134.4127&0.000000021882\\
(13,2,1)& 3&0&1465.5626&0.000000031703\\
(12,4,0)& --&--&5040.7184&0.000000596688\\
(12,3,1)& 4&0&7316.8139&0.000001088645\\
(12,2,2)& 5&0&8551.5945&0.000001460028\\
(11,5,0)& --&--&16255.7058&0.000007710948\\
(11,4,1)& 5&0&26646.9787&0.000017673415\\
(11,3,2)& 6&1&35516.8360&0.000029848564\\
(10,6,0)& --&--&38300.8686&0.000052951447\\
(10,5,1)& 6&0&71988.3451&0.000156556413\\
(10,4,2)& 8&1&108961.6972&0.000332160313\\
(10,3,3)& 8&2&124253.3641&0.000418288187\\
(9,7,0)& --&--&65227.2515&0.000183839144\\
(9,6,1)& 7&0&143152.3479&0.000733913103\\
(9,5,2)& 9&2&247939.6689&0.002008638773\\
(9,4,3)& 10&3&320986.2915&0.003177631988\\
(8,8,0)& --&--&78440.0862&0.000290035432\\
(8,7,1)& 8&0&204758.6410&0.001676299730\\
(8,6,2)& 11&2&410936.6736&0.006194731655\\
(8,5,3)& 12&4&608259.9214&0.012641673731\\
(8,4,4)& 13&4&690871.1760&0.015881023028\\
(7,7,2)& 11&3&485122.9853&0.008968430903\\
(7,6,3)& 13&5&830558.9804&0.024712243136\\
(7,5,4)& 14&6&1076870.7243&0.040046668103\\
(6,6,4)& 15&6&1244849.8786&0.054072943466\\
(6,5,5)& 15&7&1421227.6017&0.069752589222\\
\hline
\end{tabular}
\end{center}
\end{table}

While the non-convex optimization problems of Algorithm $\mathcal{B}$ 
seem intractable when studying higher
powers of~$\A$, these powers can be analyzed by applying Algorithm $\mathcal{A}$, 
as in the previous subsection,
but using this time the lower bounds on the values of the components $V_\rho(\A^{\otimes 8}(a,b,c))$ obtained by Algorithm~$\mathcal{B}$ as a starting point. This strategy can be equivalently described as using 
Algorithm $\mathcal{A}$
to analyze powers of 
$\A'$, where $\A'=\A^{\otimes 8}$, with lower bounds on the values of each component of $\A'$ computed by Algorithm~$\mathcal{B}$.
The lower bounds we obtain using this method 
for powers 16 and 32 are given in Table \ref{table:chart2} and Figure \ref{fig:graph}. They show that
$\omega<2.3728640$ and
$\omega<2.3728639$, respectively.

\begin{table}[h!]
\renewcommand\arraystretch{1}
\begin{center}
\caption{
Lower bounds on $V_\rho(\A^{\otimes 2^r})-(q+2)^{2^r}$ for $r=4$ and $r=5$,
with $q=5$, computed by performing the analysis on the second and fourth power
of $\A^{\otimes 8}$ using Algorithm $\mathcal{A}$, with lower bounds of each component of $\A^{\otimes 8}$ 
computed by Algorithm $\mathcal{B}$. Plots representing the data of this table are given 
in Figure~\ref{fig:graph}.
}\label{table:chart2}\vspace{2mm}
\begin{tabular}{|c|r|r|}
\hline
\multirow{2}{*}{$\rho$}& 
lower bound on$\:\:\:$&
lower bound on$\:\:\:$\\
& 
$V_\rho(\A^{\otimes 16})-(q+2)^{16}$&
$V_\rho(\A^{\otimes 32})-(q+2)^{32}$\\
\hline
2.3728670&8.3460\:E+08&5.7365\:E+22\\
2.3728660& 5.5802\:E+08&3.8857\:E+22\\
2.3728650&2.8110\:E+08&2.0175\:E+22\\
2.3728640& 2.5866\:E+06&3.1912\:E+21\\
2.3728639&-2.4921\:E+07&1.2306\:E+21\\
2.3728638& -5.2465\:E+07&-1.0995\:E+21\\
2.3728630&-2.7407\:E+08&-1.5750\:E+22\\
2.3728620&-5.5031\:E+08&-3.3889\:E+22\\
2.3728610&-8.2774\:E+08&-5.3088\:E+22\\
2.3728600& -1.1051\:E+09&-7.0489\:E+22\\
\hline
\end{tabular}
\end{center}
\end{table}

\begin{figure*}[ht!]
     \begin{center}
       \subfigure[Power 16]{%
            \label{fig:first}
            \includegraphics[scale=1.2]{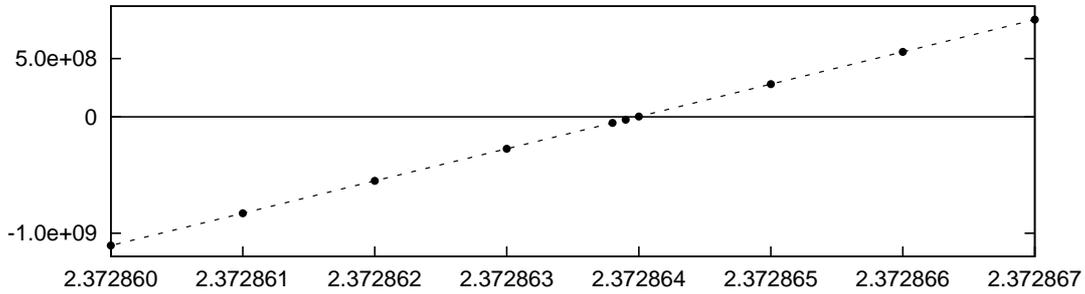}
        }
      \subfigure[Power 32]{%
            \label{fig:second}
            \includegraphics[scale=1.2]{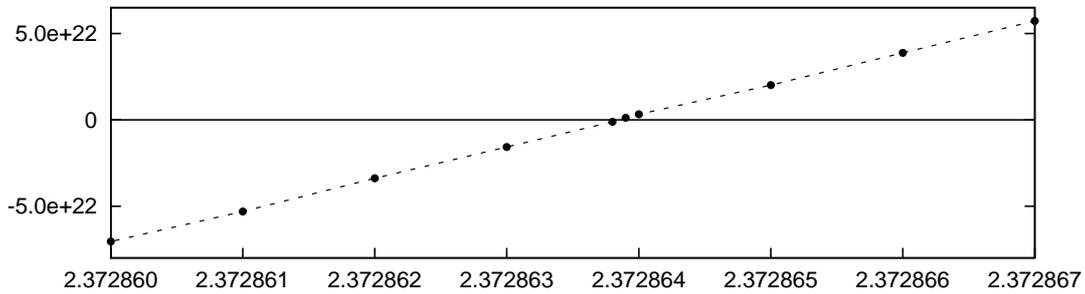}
        }
      \end{center}\vspace{-3mm}
      \caption{Graphs representing the data of Table \ref{table:chart2}.
      The horizontal axes represent $\rho$, while the vertical axes 
      represent the lower bounds on $V_\rho(\A^{\otimes 16})-(q+2)^{16}$,
      for graph (a), 
      and  $V_\rho(\A^{\otimes 32})-(q+2)^{32}$, for graph (b).}\label{fig:graph}
\end{figure*}
\clearpage
\newpage

\section*{Acknowledgments}
The author is grateful to Virginia Vassilevska Williams for helpful correspondence,
and to Harumichi Nishimura, Suguru Tamaki and Yuichi Yoshida for stimulating discussions.
This work is supported by
the Grant-in-Aid for Young Scientists~(B)~No.~24700005 of the Japan Society for the Promotion of Science
and the Grant-in-Aid for Scientific Research on Innovative Areas~No.~24106009 of
the Ministry of Education, Culture, Sports, Science and Technology in Japan.


\newpage
\appendix
\section*{APPENDIX}
\section{Proof of Theorem 4.1}
In this appendix we prove 
Theorem \ref{th:main}.
Before giving the proof in Section \ref{sub:aproof}, we present
some preliminaries in Sections \ref{sub:nvalued} and \ref{sub:comblem}.
\subsection{Rational-valued probability distributions} \label{sub:nvalued}
Since it will be convenient to deal with probability distributions with rational 
values, instead of arbitrary real values, we first introduce the following definition.
Here we use the notations of Section \ref{sec:prelim} (in particular, $S$ denotes a finite
subset of $\Int\times\Int\times\Int$ and $\alpha_1,\alpha_2,\alpha_3$ denote its coordinate
functions).
\begin{definition}
Let $N$ be a positive integer.
A distribution $P\in \D(S)$ is $\frac{1}{N}$-valued
if $P(\myvec{s})$ is a multiple of $1/N$ for all $\myvec{s}\in S$.
\end{definition}
For any positive integer $N$, any $\ell\in\{1,2,3\}$, and any $\frac{1}{N}$-valued distribution $P\in \D(S)$,
we define the quantity
\begin{align*}
\Phi_{\ell,N}(P)&=\frac{N!}{\prod_{a\in \alpha_\ell(S)}(N\cdot P_\ell(a))!},
\end{align*}
and, for any $1/N$-valued distribution $Q\in \D(S)$ compatible with $P$, 
we also define the quantity
\begin{align*}
\Phi'_{\ell,N}(P,Q)&=\prod_{a\in \alpha_\ell(S)} \frac{(N\cdot P_\ell(a))!}{\prod_{\myvec{s}\in \alpha^{-1}_\ell(a)}(N\cdot Q(\myvec{s}))!}.
\end{align*}
Note that both quantities are integers, since $\Phi_{\ell,N}(P)$ can be rewritten as a multinomial coefficient and 
$\Phi'_{\ell,N}(P,Q)$ as a product of multinomial coefficients.
The following lemma, proved in a straightforward way using Stirling's approximation, 
will also be useful.

\begin{lemma}\label{lemma:asympt}
For any $\delta>0$ there exists a value $N_0$ such that, 
for any $\ell\in\{1,2,3\}$ and any $N\ge N_0$,
the inequality
\[
\left|
H(P_\ell)-\frac{1}{N}\log{(\Phi_{\ell,N}(P))}
\right|
\le \delta
\]
holds for all $1/N$-valued distribution $P\in \D(S)$, 
and
\[
\left|
(H(Q)-H(P_\ell))-\frac{1}{N}\log\left(\Phi'_{\ell,N}(Q,P)\right)\right|\le \delta
\]
for all $1/N$-valued compatible distributions $P,Q\in \D(S)$.
\end{lemma}

\subsection{A combinatorial lemma}\label{sub:comblem}
The proof of Theorem \ref{th:main} 
will use a generalization of a combinatorial 
result proved in \cite{Coppersmith+90},
which we describe below.

Let $S$ be a finite
subset of $\Int\times\Int\times\Int$.
Let $M$ be a large integer and consider a set
\[
\Lambda\subseteq \Int^M\times \Int^M\times \Int^M
\] 
such that,
for any $(\myvec{u},\myvec{v},\myvec{w})\in \Lambda$ and
any $\ell\in\{1,\ldots,M\}$,  the element $(u_\ell,v_\ell,w_\ell)$ is in $S$.
Let $\beta_1\colon\Lambda\to \Int^M$, $\beta_2\colon\Lambda\to \Int^M$ and 
$\beta_3\colon\Lambda\to \Int^M$ be the 
three coordinate functions of $\Lambda$. 

Let $\Lambda^\ast$ be a subset of $\Lambda$.
Assume that there exist
three integers $T$, $\Nn$, $\Nn^\ast$ such that,
for each $\ell\in \{1,2,3\}$, the following three conditions hold:
\begin{itemize}
\item[(i)]
$|\beta_\ell(\Lambda)|=T$ and $|\beta_\ell(\Lambda^\ast)|=T$,
\item[(ii)]
$|\beta^{-1}_\ell(\beta_\ell(\myvec{u},\myvec{v},\myvec{w}))|=\Nn$ for all 
$(\myvec{u},\myvec{v},\myvec{w})\in \Lambda$,
\item[(iii)]
$|\beta^{-1}_\ell(\beta_\ell(\myvec{u},\myvec{v},\myvec{w}))\cap\Lambda^\ast|=\Nn^\ast$ for all 
$(\myvec{u},\myvec{v},\myvec{w})\in \Lambda^\ast$.
\end{itemize}
Condition (i) implies that $\beta_\ell(\Lambda^\ast)=\beta_\ell(\Lambda)$.
Condition (ii) means that $\beta_\ell$ is an $\Nn$-to-one map 
from $\Lambda$ to $\beta_\ell(\Lambda)$, 
while 
Condition~(iii) means that $\beta_\ell$ is an $\Nn^\ast$-to-one map from $\Lambda^\ast$ to $\beta_\ell(\Lambda^\ast)$.
The following lemma, which has been proven in \cite{Davie+13,LeGallFOCS12}, 
will be crucial for our analysis.
It essentially states that a large subset $\Delta\subseteq \Lambda^\ast$ can be 
derived from this construction,
such that each $\beta_\ell$ is a one-to-one map
from $\Delta$ to $\beta_\ell(\Delta)$.

\begin{lemma}\label{lemma:SS}
If $\Lambda$ and $\Lambda^\ast$ satisfy Conditions (i)-(iii) above, then
there exists three sets $F_1\subseteq \beta_1(\Lambda)$, $F_2\subseteq \beta_2(\Lambda)$ and $F_3\subseteq \beta_3(\Lambda)$
such that the set 
\[
\Delta=(\beta_1^{-1}(F_1)\times \beta_2^{-1}(F_2)\times \beta_3^{-1}(F_3))\cap \Lambda
\]
is a subset of $\Lambda^\ast$ and satisfies the following property for each $\ell\in\{1,2,3\}$:
\begin{equation}\label{eq4}
|\beta^{-1}_\ell(\beta_\ell(\myvec{u},\myvec{v},\myvec{w}))\cap\Delta|=1 \hspace{2mm}\textrm{for all } 
(\myvec{u},\myvec{v},\myvec{w})\in \Delta.
\end{equation}
Moreover, 
for any $\epsilon>0$, if $M$ is taken large enough then 
$\Delta$  is such that
\[
|\Delta|\ge c\times \frac{T\Nn^\ast}{\Nn^{1+\epsilon}},
\]
where $c$ is a constant depending only on $\epsilon$.
\end{lemma}

\subsection{Proof of Theorem 4.1}\label{sub:aproof}
We are now ready to give the proof of the theorem.

\begin{proof}[Proof of Theorem \ref{th:main}]
We will denote $\alpha_1,\alpha_2,\alpha_3\colon S\to \Int$ the three coordinate functions of $S$, 
as in Section~\ref{sec:prelim}, and write
$I=\alpha_1(S)$,  $J=\alpha_2(S)$ and $K=\alpha_3(S)$

Assume that the distribution $P$ has rational values. 
This assumption can be done without loss of generality since~$P$ can be approximated with arbitrary precision by a probability distribution with rational values.
We will
take an integer $N$ large enough so that, for each $\myvec{s}\in S$, the value $N\cdot P(\myvec{s})$ is an integer
(i.e.,~$P$ is $\frac{1}{N}$-valued).

Consider the trilinear form $t^{\otimes N}$. The decomposition $D$ of $t$
induces a decomposition of this trilinear form with support 
\begin{align*}
\supp(t^{\otimes N})=
\Big\{
(\myvec{a},\myvec{b},\myvec{c})\in I^N\!\!\times\! J^N\!\!\times\! K^N
\:|\:(a_\ell,&b_\ell,c_\ell)\in S \textrm{ for all }
\ell\in\{1,\ldots,N\}
\Big\},
\end{align*}
and components 
\[
t^{\otimes N}(\myvec{a},\myvec{b},\myvec{c})=\bigotimes_{\ell=1}^N t(a_\ell,b_\ell,c_\ell).
\] 
For any element $(\myvec{a},\myvec{b},\myvec{c})$ in $\supp(t^{\otimes N})$,
we define the type of $(\myvec{a},\myvec{b},\myvec{c})$
as the $\frac{1}{N}$-valued probability distribution $Q$ on $S$ such that
\[
Q(\myvec{s})=
\frac{|\{\ell \in \{1,\ldots,N\}\:|\: (a_\ell,b_\ell,c_\ell)=\myvec{s}\}|}{N}
\]
for all $\myvec{s}\in S$. 

For any $\myvec{a}\in I^N$, we say that 
$\myvec{a}$ has distribution~$P_1$ if 
\[
|\{\ell \in \{1,\ldots,N\}\:|\: a_\ell=i\}|=N\cdot P_1(i)\hspace{2mm}\textrm{ for all } i\in I.
\]
Similarly, we say that $\myvec{b}\in J^N$ has distribution~$P_2$ if  
\[
|\{\ell \in \{1,\ldots,N\}\:|\: b_\ell=j\}|=N\cdot P_2(j)\hspace{2mm}\textrm{ for all } j\in J,
\]
and say that $\myvec{c}\in K^N$ has distribution~$P_3$ if 
\[
|\{\ell \in \{1,\ldots,N\}\:|\: c_\ell=k\}|=N\cdot P_3(k)\hspace{2mm}\textrm{ for all } k\in K,
\]
respectively.
Define the sets 
\begin{align*}
\mset{S}_1&=\{\myvec{a}\in I^N\:|\: \myvec{a} \textrm{ has distribution }P_1\},\\
\mset{S}_2 &=\{\myvec{b}\in J^N\:|\: \myvec{b} \textrm{ has distribution } P_2\},\\
\mset{S}_3 &=\{\myvec{c}\in K^N\:|\: \myvec{c} \textrm{ has distribution } P_3\},\\
\mset{S}&=(\mset{S}_1\times \mset{S}_2\times \mset{S}_3)\cap \supp(t^{\otimes N}).
\end{align*}
We have $|\mset{S}_\ell|=\Phi_{\ell,N}(P)$ for each $\ell\in\{1,2,3\}$.
We now perform a first pruning. For each $(\myvec{a},\myvec{b},\myvec{c})\in \supp(t^{\otimes N})\setminus \mset{S}$, 
we do as follows: if $\myvec{a}\notin \mset{S}_1$, then we set to zero all the $x$-variables appearing in 
$t(\myvec{a},\myvec{b},\myvec{c})$; if $\myvec{b}\notin \mset{S}_2$, then we set to zero all the $y$-variables appearing in 
$t(\myvec{a},\myvec{b},\myvec{c})$; if $\myvec{c}\notin \mset{S}_3$, then we set to zero all the $z$-variables appearing in 
$t(\myvec{a},\myvec{b},\myvec{c})$.
After this pruning, the remaining form is
\begin{equation}\label{eq:dec}
\mset{t}=\sum_{(\myvec{a},\myvec{b},\myvec{c})\in \mset{S}}t^{\otimes N}(\myvec{a},\myvec{b},\myvec{c}).
\end{equation}
Note that Eq.~(\ref{eq:dec}) actually represents a decomposition of $\mset{t}$ with support~$\mset{S}$.
Up to this point, we have shown that $\mset{t}\degen t^{\otimes N}$.
 
Observe that, for any $(\myvec{a},\myvec{b},\myvec{c})\in \mset{S}$, 
the number of couples $(\myvec{b'},\myvec{c'})\in J^N\times K^N$ such that $(\myvec{a},\myvec{b'},\myvec{c'})\in \mset{S}$ is 
independent of~$\myvec{a}$. This number is thus equal to $|\mset{S}|/|\mset{S}_1|$, and will be denoted $\Nn_1$. Similarly, again 
for any $(\myvec{a},\myvec{b},\myvec{c})\in \mset{S}$, the number of elements of $\mset{S}$ with second coordinate $\myvec{b}$
is $\Nn_2=|\mset{S}/|\mset{S}_2|$, and the number of elements of $\mset{S}$ with third coordinate $\myvec{c}$
is $\Nn_3=|\mset{S}|/|\mset{S}_3|$. 

For any $\myvec{a}\in \mset{S}_1$, the number of couples $(\myvec{b},\myvec{c})$, 
such that $(\myvec{a},\myvec{b},\myvec{c})$ is in $\mset{S}$ and $(\myvec{a},\myvec{b},\myvec{c})$ has type $Q$ is
exactly $\Phi'_{1,N}(P,Q)$ if $Q$ is a $1/N$-valued distribution on $S$ compatible with $P$, and zero otherwise.
Thus
\[
\Nn_1=\sum_{Q} \Phi'_{1,N}(P,Q),
\]
where the sum is over all $1/N$-valued distributions $Q$ compatible with~$P$. 
Similarly, $\Nn_2=\sum_{Q} \Phi'_{2,N}(P,Q)$ and $\Nn_3=\sum_{Q} \Phi'_{3,N}(P,Q)$.
Observe that, if $(\myvec{a},\myvec{b},\myvec{c})$ is in $\mset{S}$ and $(\myvec{a},\myvec{b},\myvec{c})$ has type $Q$,
then 
\begin{equation}\label{eq8}
\begin{split}
V_\rho(t^{\otimes N}(\myvec{a},\myvec{b},\myvec{c}))&\ge \prod_{\ell=1}^n V_\rho(t(a_\ell,b_\ell,c_\ell))\\
&= \prod_{\myvec{s}\in S} [V_\rho(t(\myvec{s}))]^{N\cdot Q(\myvec{s})},
\end{split}
\end{equation}
where the inequality is obtained using the supermultiplicativity of the value.

Now consider the trilinear form $\mset{t}\otimes \mset{t}_{\mathsf C}\otimes \mset{t}_{\mathsf C^2}$. Let
\[
\Lambda\subseteq \Int^{3N}\times \Int^{3N}\times \Int^{3N}
\]
be the set of all triples
$(\myvec{u},\myvec{v},\myvec{w})$ with $\myvec{u}=(\myvec{a},\myvec{b'},\myvec{c''})\in \mset{S}_1\times \mset{S}_2\times \mset{S}_3$,
$\myvec{v}=(\myvec{b},\myvec{c'},\myvec{a''})\in \mset{S}_2\times \mset{S}_3\times \mset{S}_1$ and
$\myvec{w}=(\myvec{c},\myvec{a'},\myvec{b''})\in \mset{S}_3\times \mset{S}_1\times \mset{S}_2$
satisfying the conditions
\[
(\myvec{a},\myvec{b},\myvec{c})\in \mset{S}, \:\:\:(\myvec{a'},\myvec{b'},\myvec{c'})\in \mset{S}, \:\:\:(\myvec{a''},\myvec{b''},\myvec{c''})\in \mset{S}.
\]
The decomposition of $\mset{t}$ of Eq.~(\ref{eq:dec}) naturally induces a decomposition of
$\mset{t}\otimes \mset{t}_{\mathsf C}\otimes \mset{t}_{\mathsf C^2}$ with support $\Lambda$.
In this decomposition, 
the component associated with a triple $(\myvec{u},\myvec{v},\myvec{w})\in\Lambda$,
which we denote $\hat t(\myvec{u},\myvec{v},\myvec{w})$,
 is
\begin{equation}\label{eq7}
\hat t(\myvec{u},\myvec{v},\myvec{w})=
t^{\otimes N}(\myvec{a},\myvec{b},\myvec{c})\otimes 
t^{\otimes N}_{\mathsf C}(\myvec{b'},\myvec{c'},\myvec{a'}) \otimes 
t^{\otimes N}_{\mathsf C^2}(\myvec{c''},\myvec{a''},\myvec{b''}).
\end{equation}
Let $\beta_1,\beta_2,\beta_3\colon\Lambda\to \Int^{3N}$ be the three coordinate functions of~$\Lambda$:
for any $(\myvec{u},\myvec{v},\myvec{w})\in\Lambda$,
\begin{align*}
\beta_1(\myvec{u},\myvec{v},\myvec{w})&=\myvec{u},\:\:\:
\beta_2(\myvec{u},\myvec{v},\myvec{w})=\myvec{v},\:\:\:
\beta_3(\myvec{u},\myvec{v},\myvec{w})=\myvec{w}.
\end{align*}
We have
$
|\beta_\ell(\Lambda)|=T 
$
for each $\ell\in\{1,2,3\}$,
where
\[
T=|\mset{S}_1|\times |\mset{S}_2|\times |\mset{S}_3|=
\Phi_{1,N}(P)\Phi_{2,N}(P)\Phi_{3,N}(P).
\]
Moreover, for any $\ell\in\{1,2,3\}$ and any $(\myvec{u},\myvec{v},\myvec{w})\in\Lambda$,
we have
$|\beta_\ell^{-1}(\beta_\ell(\myvec{u},\myvec{v},\myvec{w}))|=\Nn$, with
\begin{align*}
\Nn&=\Nn_1\Nn_2\Nn_3\\
&=\sum_{Q,Q',Q''} \Phi'_{1,N}(P,Q)\Phi'_{2,N}(P,Q')\Phi'_{3,N}(P,Q''),
\end{align*}
where the sum is over all triples of $1/N$-valued distributions $Q$, $Q'$ and $Q''$ compatible with $P$. 

Let $\Lambda^\ast$ be the subset of $\Lambda$ containing all those 
$(\myvec{u},\myvec{v},\myvec{w})$ such that 
$(\myvec{a},\myvec{b},\myvec{c})$, $(\myvec{a'},\myvec{b'},\myvec{c'})$ and $(\myvec{a''},\myvec{b''},\myvec{c''})$ are of type~$P$.
For any $\ell\in\{1,2,3\}$ and any $(\myvec{u},\myvec{v},\myvec{w})\in\Lambda^\ast$,
we have
$|\beta_\ell^{-1}(\beta_\ell(\myvec{u},\myvec{v},\myvec{w}))\cap \Lambda^\ast|=\Nn^\ast$, where
\[
\Nn^\ast=\Phi'_{1,N}(P,P)\Phi'_{2,N}(P,P)\Phi'_{3,N}(P,P).
\]

We can now use Lemma \ref{lemma:SS}:
for any $\epsilon>0$, if $N$ is taken large enough
there exists a set $\Delta\subseteq \Lambda^\ast$
of size 
\begin{equation}\label{eq9}
|\Delta|\ge c\times \frac{T\Nn^\ast}{\Nn^{1+\epsilon}},
\end{equation}
for some constant $c$ depending only on $\epsilon$,
such that the form 
\begin{equation}\label{eq10}
\hat t=\bigoplus_{(\myvec{u},\myvec{v},\myvec{w})\in\Delta}
\hat t(\myvec{u},\myvec{v},\myvec{w}).
\end{equation}
can be obtained from 
$\mset{t}\otimes \mset{t}_{\mathsf C}\otimes \mset{t}_{\mathsf C^2}$
by zeroing $x$-variables, $y$-variables and 
$z$-variables (which means in particular that 
the form $\hat t$ is a degeneration 
$\mset{t}\otimes \mset{t}_{\mathsf C}\otimes \mset{t}_{\mathsf C^2}$).
Note that the sum is direct since the components $\hat t(\myvec{u},\myvec{v},\myvec{w})$ do not 
share variables, due to the property of Eq.~(\ref{eq4}) in Lemma~\ref{lemma:SS}.
Up to this point, we have thus shown that
\[
\hat t\degen\mset{t}\otimes \mset{t}_{\mathsf C}\otimes \mset{t}_{\mathsf C^2}\degen (t\otimes t_{\mathsf C} \otimes t_{\mathsf C^2})^{\otimes N},
\]
which implies that
\[
V_\rho(t)\ge \left(V_\rho(\hat t)\right)^{\frac{1}{3N}}.
\]

From Eqs.~(\ref{eq8}) and (\ref{eq7}), and from the definition of $\Lambda^\ast$, we know that
\[
V_\rho(\hat t(\myvec{u},\myvec{v},\myvec{w}))\ge
\Bigg(
\prod_{\myvec{s}\in S} [V_\rho(t(\myvec{s}))]^{N\cdot P(\myvec{s})}
\Bigg)^3
\]
for any $(\myvec{u},\myvec{v},\myvec{w})\in\Lambda^\ast$,
since the value is supermultiplicative and invariant under permutation of coordinates.
From Ineq.~(\ref{eq9}) and Eq.~(\ref{eq10}),  we conclude that 
\[
V_\rho(\hat t)\ge
\frac{c\Nn^\ast}{\Nn^{1+\epsilon}}\times \prod_{\ell=1}^N\Phi_{\ell,N}(P)\times \prod_{\myvec{s}\in S} [V_\rho(t(\myvec{s}))]^{3N\cdot P(\myvec{s})},
\]
since the value is superadditive. By using Lemma \ref{lemma:asympt} to approximate each $\Phi_{\ell,N}(P)$,
then taking the power $1/(3N)$, and finally taking the logarithm,
we conclude that, for any $\delta>0$, if $N$ is taken large enough then the following inequality holds:
\[
\log(V_\rho(t))\ge \frac{\sum_{\ell=1}^3H(P_\ell)}{3}+\sum_{\myvec{s}\in S} P(\myvec{s})\log (V_\rho(t(\myvec{s})))-\delta-\zeta,
\]
where $\zeta=\frac{1}{3N}\log (\frac{\Nn^{1+\epsilon}}{c\Nn^\ast})$. 

Finally, note that,
since the number of $1/N$-valued probability distributions on~$S$ is 
at most $(N+1)^{|S|}$, we have
\begin{align*}
\Nn&\le (N+1)^{3|S|}\times \prod_{\ell=1}^N \max_{Q}\left[\Phi'_{\ell,N}(P,Q)\right],
\end{align*}
where each maximum is over all $1/N$-valued distributions $Q$ over $S$ compatible with $P$.
Using this inequality, and Lemma~\ref{lemma:asympt} combined with the triangular inequality, we conclude that 
\begin{align*}
\zeta&\le \frac{3|S|(1+\epsilon)\log(N+1)-\log c}{3N}+(1+\epsilon)\left(\max_Q\left[H(Q)\right]-\sum_{\ell=1}^3 \frac{H(P_\ell)}{3}\right)\\
&\hspace{56mm}-\left(H(P)-\sum_{\ell=1}^3 \frac{H(P_\ell)}{3}\right) +(2+\epsilon)\delta\\
&\le\Gamma_S(P)+\frac{3|S|(1+\epsilon)\log(N+1)-\log c}{3N}+\epsilon\left(\max_Q\left[H(Q)\right]+\delta\right)+2\delta.
\end{align*}
Note that $H(Q)\le \log|S|$ for any distribution $Q$ on $S$. 
By letting $\delta$ and $\epsilon$ go to zero, and 
letting $N$ go to infinity, 
we obtain the inequality 
\begin{align*}
\log(V_\rho(t))\ge&
\sum_{\ell=1}^3\frac{H(P_\ell)}{3}+\sum_{\myvec{s}\in S} P(\myvec{s})\log (V_\rho(\myvec{s}))-\Gamma_S(P),
\end{align*}
as claimed.
\end{proof}

\section{Another Application}
In this appendix we apply our theory to another trilinear 
form, that we denote $\B$, which was 
proposed by Coppersmith and Winograd 
in Section~9 of \cite{Coppersmith+90} and is asymmetric (i.e., the decomposition of $\B$ 
is not $\mathbb{S}_3$-invariant).
Coppersmith and Winograd showed how the analysis of $\B$
gives the bound $\omega<2.46015$.
The upper bounds on~$\omega$ we obtain from the analysis of the powers of 
$\B$ described in this section are summarized in Table~\ref{table:chart3}.
They are worse than the upper bounds on $\omega$ obtained from~$\A$, but
they illustrate how our techniques work when studying asymmetric forms.
All the programs used 
for the computations presented here are available
at~\cite{files}.

\begin{table}[h!]
\renewcommand\arraystretch{1}
\begin{center}
\caption{Upper bounds on $\omega$ obtained by studying the $m$-th power of the 
construction $\B$.}\label{table:chart3}\vspace{3mm}
\begin{tabular}{|c|l|}
\hline
$m$& Upper bound on $\omega$\\ 
\hline
1& $\omega<2.46015$ (see \cite{Coppersmith+90})\\
\hline
2&  $\omega<2.44998$ \\
\hline
4&  $\omega<2.44303$ \\
\hline
8&  $\omega<2.44278$ 
\tabularnewline\hline
\end{tabular}
\end{center}
\end{table}

\subsection{Construction}
Let $\field$ be an arbitrary field.
Let $q$ be a positive integer, and consider three 
vector spaces $U$, $V$ and $W$ of dimension $q+1$ over $\field$. 
Take 
a basis $\{x_1,\ldots,x_{q+1}\}$ of $U$, 
a basis $\{y_0,\ldots,y_{q}\}$ of $V$, 
a basis $\{z_0,\ldots,z_{q}\}$ of $W$. 

Consider the following trilinear form $\B$ on
$(U,V,W)$:
\[
\B=x_{q+1}y_0z_0+\sum_{i=1}^q(x_iy_iz_0+x_iy_0z_i).
\]
Coppersmith and Winograd \cite{Coppersmith+90} showed
that $\underline R(\B)=q+1$.
Consider the following decomposition of $U$, $V$ and $W$:
\[
U=U_1\oplus U_2,\:\:\: V=V_0\oplus V_1,\:\:\: W=W_0\oplus W_1,
\]
where
$U_1=\myspan{x_1,\ldots,x_q}$ and $U_2=\myspan{x_{q+1}}$,
$V_0=\myspan{y_0}$ and $V_1=\myspan{y_1,\ldots,y_{q}}$,
$W_0=\myspan{z_0}$ and $W_1=\myspan{z_1,\ldots,z_{q}}$.
This decomposition induces a decomposition $D$ of $\B$ with
tight
support
\[
S=\{(2,0,0),(1,1,0),(1,0,1)\}
\]
 and components
\begin{align*}
\B(2,0,0)&=x_{q+1}y_0z_0\cong\braket{1,1,1},\\
\B(1,1,0)&=\sum_{i=1}^qx_iy_iz_0\cong\braket{1,q,1},\\
\B(1,0,1)&=\sum_{i=1}^qx_iy_0z_i\cong\braket{q,1,1}.
\end{align*}
For any $\rho\in[2,3]$, we have $V_\rho(\B(2,0,0))=1$ and
$V_\rho(\B(1,1,0))=V_\rho(\B(1,0,1))\ge q^{\rho/3}$.

We now use Theorem \ref{th:main} to obtain an upper bound on~$\omega$.
Let $P$ be a probability distribution in $\D(S)$.
Let us write
$P(2,0,0)=a_1$, $P(1,1,0)=a_2$ and $P(1,0,1)=a_3$. 
The marginal distributions of $P$ are
$\myvec{P_1}=(a_1,a_2+a_3)$, $\myvec{P_2}=(a_2,a_1+a_3)$
and $\myvec{P_3}=(a_3,a_1+a_2)$.
Since the only probability distribution on $S$ compatible with $P$ is 
$P$, we have $\Gamma_S(P)=0$. 
Theorem~\ref{th:main} 
thus gives
\[
V_\rho(\B)\ge
\frac
{q^{(a_2+a_3)\rho/3}}
{
\left[
a_1^{a_1}
a_2^{a_2}
a_3^{a_3}
(a_1+a_2)^{a_1+a_2}
(a_1+a_3)^{a_1+a_3}
(a_2+a_3)^{a_2+a_3}
\right]^{1/3}
},
\]
for any $\rho\in[2,3]$.
Evaluating this expression with $q=4$, $a_1=0.0302$, $a_2=a_3=0.4849$
and $\rho=2.46015$ gives $V_\rho(\B)> 5.0000005$.
Using Theorem \ref{th:value} and the fact that $\underline R(\B)=q+1$, 
we conclude that $\omega<2.46015$. This is the same upper bound as the bound
found in Section 9 of \cite{Coppersmith+90}.

\subsection{Analyzing the powers using Algorithm $\mathcal{A}$}
For any $r\ge 1$, we now consider the tensor $\B^{\otimes 2^r}$ and analyze 
it using the framework of Section \ref{section:powers}.
The support of its decomposition $D^{2^r}$ is the set of all triples
\[
(a,b,c)\in\{2^r,\ldots,2^{r+1}\}\times \{0,\ldots,2^{r}\}\times\{0,\ldots,2^{r}\}
\]
such that $a+b+c=2^{r+1}$.
Note that the decomposition~$D$ of $\B$ satisfies the conditions
 of Lemma~\ref{lemma:sym2}
for the subgroup $L=\{\identi,(2\:3)\}$ of $\mathbb{S}_3$, which implies that 
$D^{2^r}$ is $\{\identi,(2\:3)\}$-invariant. Thus we only need to consider probability 
distributions in $\D(\supp(\B^{\otimes 2^r}),\{\identi,(2\:3)\})$, i.e., the probability distributions 
$P\in\D(\supp(\B^{\otimes 2^r}))$ 
such that 
\[
P(a,b,c)=P(a,c,b) \textrm{ for all } (a,b,c)\in\supp(\B^{\otimes 2^r}).
\]
This set of probability distributions can be parametrized by 
\[
\dim(\F(\supp(\B^{\otimes 2^r}),\{\identi,(2\:3)\}))
\]
parameters. 
We also need 
a lower bound on the value of each component $\B^{\otimes 2^r}(a,b,c)$
before applying $\mathcal{A}$ on $t^{\otimes 2^r}$. These lower bounds 
are computed recursively 
by applying Algorithm~$\mathcal{A}$ on the decomposition 
$S^{2^r}_{abc}$ given in~Section~\ref{section:powers}. 
Actually, we do not need to apply~$\mathcal{A}$ when $a=0$, $b=0$ or $c=0$,
since a lower bound on the value can be found analytically in this case,
as stated in the following lemma.
\begin{lemma}\label{lemma:value-s}
For any $r\ge 0$ and any $b\in\{0,1,\ldots,2^r\}$, 
\[
V_\rho\big(\B^{\otimes 2^r}(2^{r+1}-b,b,0)\big)\ge \big({2^r \choose b}q^b\big)^{\rho/3}.
\]
\end{lemma}
\begin{proof}
It is straightforward to check, by recurrence on~$r$, that 
\[
\B^{\otimes 2^r}(2^{r+1}-b,b,0)\cong \braket{1,{2^r \choose b}q^b,1}
\]
for any $r\ge 0$ and any $b\in\{0,1,\ldots,2^r\}$. We can then use the definition of the value.
\end{proof}
Table \ref{table:A-param} presents, for $r\in\{1,2,3,4,5\}$,
the number of variables in the global optimization problem
and the associated 
compatibility degree. We show below how to solve these
optimizations problems using Algorithm $\mathcal{A}$ for 
$r=1,2,3$.

\begin{table}[h!]
\renewcommand\arraystretch{1}
\begin{center}
\caption{Parameters for the analysis of $\B^{\otimes 2^r}$.}\label{table:A-param}\vspace{2mm}
\begin{tabular}{|c|c|c|}
\hline
$\!r\!$&$\dim(\F(\supp(\B^{\otimes 2^r}),\{\identi,(2\:3)\}))$&$\chi(\supp(\B^{\otimes 2^r}),\{\identi,(2\:3)\})$\bigstrut\\ 
\hline
\!1\!& 4&0\\
\!2\!& 9&1\\
\!3\!& 25&9\\
\!4\!& 81&49\\
\!5\!& 289&225\\
\hline
\end{tabular}
\end{center}
\end{table}

The case $r=1$ is easy to deal with, since the compatibility degree is zero.
Note that lower bounds on the
values of all components but one can be computed directly using Lemma~\ref{lemma:value-s},
as summarized in Table~\ref{table:A-power2}. A lower bound on the value of the component $\B^{\otimes 2}(2,1,1)$ is computed
using its decomposition $D^2_{211}$, which has support $S^2_{211}=\{(1,1,0),(1,0,1)\}$: we take the probability 
distribution on $S^2_{211}$ that assigns probability 1/2 to both $(1,1,0)$ and $(1,0,1)$, and,
from Theorem \ref{th:main}, we obtain the lower bound
\[
V_\rho(\B^{\otimes 2}(2,1,1))\ge \frac{(q^{\rho/3}\times q^{\rho/3})^{\frac{1}{2}+\frac{1}{2}}}{\left[1\times \frac{1}{2}\times \frac{1}{2}\right]^{1/3}}=2^{2/3}q^{2\rho/3}.
\] 
Let $P$ be a probability distribution in $\D(\supp(\B^{\otimes 2}),\{\identi,(2\:3)\})$, and 
write $P(4,0,0)=a_1$, $P(3,1,0)=a_2$, $P(2,2,0)=a_3$ and
$P(2,1,1)=a_4$. 
The partial distributions are
$\myvec{P_1}=(a_1,2a_2,2a_3+a_4)$
and
$\myvec{P_2}=\myvec{P_3}=(a_3,a_2+a_4,a_1+a_2+a_3)$.
Since the compatibility degree is zero, 
we have $\Gamma_{\supp(\B^{\otimes 2})}(P)=0$.
Theorem \ref{th:main} 
implies that 
\[
V_\rho(\B^{\otimes 2})\ge
\frac
{(2q)^{2a_2\rho/3}q^{2(2a_3+a_4)\rho/3}}
{
\left[
H(P_1)H(P_2) H(P_3)
\right]^{1/3}
}.
\]
Evaluating the right-hand side with $q=3$, $a_1=0.00299$, $a_2=0.06888$,
$a_3=0.21536$, $a_4=0.42853$
and $\rho=2.44998$ gives $V_\rho(\B^{\otimes 2})> 16.00002$.
Using Theorem \ref{th:value} and the fact that $\underline R(t^{\otimes 2})\le (q+1)^2$, 
we conclude that $\omega<2.44998$.
\begin{table}[h!]
\renewcommand\arraystretch{1}
\begin{center}
\caption{Components of $\B^{\otimes 2}$ and lower bounds on their values.}\label{table:A-power2}\vspace{3mm}
\begin{tabular}{|c|c|c|}
\hline
\multirow{2}{*}{$(a,b,c)$}&\multirow{2}{*}{$\B^{\otimes 2}(a,b,c)$}&lower bound on\\
&&$V_\rho(\B^{\otimes 2}(a,b,c))$\\ 
\hline
(4,0,0)& $\B(2,0,0)\otimes \B(2,0,0)$&$1$\\
\hline
(3,1,0)& $\B(2,0,0)\otimes \B(1,1,0)+\B(1,1,0)\otimes \B(2,0,0)$&$(2q)^{\rho/3}$\\
\hline
(2,2,0)& $\B(1,1,0)\otimes \B(1,1,0)$&$q^{2\rho/3}$\\
\hline
(2,1,1)& $\B(1,1,0)\otimes \B(1,0,1)+\B(1,0,1)\otimes \B(1,1,0)$&$2^{2/3}q^{2\rho/3}$\\
\hline
\end{tabular}
\end{center}
\end{table}

For $r>1$, we use Algorithm $\mathcal{A}$ to compute lower bounds on
the value of each component of $\B^{\otimes 2^r}$, and then this algorithm again to compute
a lower bound on $V_\rho(\B^{\otimes 2^r})$.
For instance, for $r=2$, the lower bounds on the values of the components we obtained, 
for $\rho=2.44303$ and $q=3$, are given in Table~\ref{table:A-power4}. 
For the choice $\hat Q\in \D(\supp(\B^{\otimes 4}),\{\identi,(2\:3)\})$ described in this table, we obtain
\[
V_{2.44303}(\B^{\otimes 4})> 256.00084>(q+1)^4,
\]
 which gives the upper bound $\omega<2.44303$. 
For $r=3$ the lower bounds on the values of the components we obtained, 
for $\rho=2.44278$ and $q=3$, are given in Table~\ref{table:A-power8}.
For the choice $\hat Q\in \D(\supp(\B^{\otimes 8}),\{\identi,(2\:3)\})$ described in this table, we obtain
\[
V_{2.44278}(\B^{\otimes 8})> 65537.50128>(q+1)^8,
\] 
which gives the upper bound $\omega<2.44278$.

\begin{table}[h!]
\renewcommand\arraystretch{1}
\begin{center}
\caption{Values of the components and optimal probability distributions for $\B^{\otimes 4}$, 
computed with Algorithm $\mathcal{A}$
for $\rho=2.44303$ and $q=3$. 
In this table, $d$ represents $\dim(\F(S^{4}_{abc},\{\identi,\pi\}))$ 
and $\chi$ represents $\chi(S^4_{abc},\{\identi,\pi\})$,
where $\pi$ is defined in Lemma \ref{lemma:sym3}.}\label{table:A-power4}\vspace{3mm}
\begin{tabular}{|c|c|c|r|c|c|}
\hline
\multirow{2}{*}{$\!\!(a,b,c)\!\!$}& 
\multirow{2}{*}{$\!d\!$}& 
\multirow{2}{*}{$\!\chi\!$}& 
$\!\!$lower bound on$\!\!$&
\multirow{2}{*}{$\hat P(a,b,c)$}&
\multirow{2}{*}{$\hat Q(a,b,c)$}\\ 
&&&$\!\!V_\rho(\B^{\otimes 8}(a,b,c))\!\!$&&\\
\hline
\!\!(8,0,0)\!\!&\!--\!&\!--\!&1.000000&    0.00000628&0.00000628\\
\!\!(7,1,0)\!\!&\!--\!&\!--\!&7.565263&    0.00049857&0.00049857\\
\!\!(6,2,0)\!\!&\!--\!&\!--\!&25.749174&  0.01413336&0.00945103\\
\!\!(6,1,1)\!\!&\!2\!&\!0\!&31.821055&  0.00000000&0.00936465\\
\!\!(5,3,0)\!\!&\!--\!&\!--\!&45.279860&  0.04378854&0.04847087\\
\!\!(5,2,1)\!\!&\!2\!&\!0\!&74.273471&  0.08053586&0.07585354\\
\!\!(4,4,0)\!\!&\!--\!&\!--\!&35.823008& 0.03895372&0.03895372\\
\!\!(4,3,1)\!\!&\!1\!&\!0\!&90.268324&  0.18465128&0.17996896\\
\!\!(4,2,2)\!\!&\!2\!&\!0\!&118.284967&0.27487107&0.28423572\\
\hline
\end{tabular}
\end{center}
\end{table}

\begin{table}[h!]
\renewcommand\arraystretch{1}
\begin{center}
\caption{Values of the components and optimal probability distributions for $\B^{\otimes 8}$, 
computed with Algorithm $\mathcal{A}$ for $\rho=2.44278$ and $q=3$. 
In this table, $d$ represents $\dim(\F(S^{8}_{abc},\{\identi,\pi\}))$ 
and $\chi$ represents $\chi(S^8_{abc},\{\identi,\pi\})$,
where $\pi$ is defined in Lemma \ref{lemma:sym3}.}\label{table:A-power8}\vspace{3mm}
\begin{tabular}{|c|c|c|r|c|}
\hline
\multirow{2}{*}{$\!(a,b,c)\!$}& 
\multirow{2}{*}{$d$}& 
\multirow{2}{*}{$\chi$}& 
lower bound on&
\multirow{2}{*}{$\hat Q(a,b,c)$}\\ 
&&&$V_\rho(\B^{\otimes 8}(a,b,c))$&\\
\hline
(16,0,0)& --&--&1.000000&0.000000000015\\
(15,1,0)& --&--&13.299974&0.000000005812\\
(14,2,0)& --&--&90.232794&0.000000471955\\
(14,1,1)& 2&0&100.217177&0.000000430768\\
(13,3,0)& --&--&388.134137&0.000014910193\\
(13,2,1)& 3&0&494.786574&0.000016086705\\
(12,4,0)& --&--&1138.656007&0.000232464978\\
(12,3,1)& 4&0&1646.513078&0.000292130426\\
(12,2,2)& 5&0&1938.340810&0.000345494424\\
(11,5,0)& --&--&2322.647734&0.001698208436\\
(11,4,1)& 4&0&3892.471176&0.002751596947\\
(11,3,2)& 5&1&5277.740474&0.003790333325\\
(10,6,0)& --&--&3231.226603&0.005184942648\\
(10,5,1)& 3&0&6556.194763&0.012338195249\\
(10,4,2)& 5&0&10275.997802&0.021913786571\\
(10,3,3)& 5&1&11840.502586&0.025524148318\\
(9,7,0)& --&--&2850.068662&0.005916530615\\
(9,6,1)& 2&0&7506.012822&0.023247197472\\
(9,5,2)& 3&0&13971.351951&0.060638752558\\
(9,4,3)& 4&0&18619.666989&0.091066196486\\
(8,8,0)& --&--&1282.348367&0.001028321368\\
(8,7,1)& 1&0&5129.393466&0.012275633494\\
(8,6,2)& 2&0&11824.370890&0.052871224881\\
(8,5,3)& 2&0&18770.018800&0.116611193180\\
(8,4,4)& 3&0&21780.677417&0.150353559875\\
\hline
\end{tabular}
\end{center}
\end{table}

\end{document}